\documentclass{article}
\usepackage[T1]{fontenc}
\usepackage[utf8]{inputenc}
\usepackage{amsmath}
\usepackage{amsthm}
\usepackage{amssymb}
\usepackage[dvipdfm]{geometry}
\usepackage[pdfusetitle,
 bookmarks=true,bookmarksnumbered=false,bookmarksopen=false,
 breaklinks=false,pdfborder={0 0 1},backref=false,colorlinks=false]
 {hyperref}

\makeatletter

\usepackage{stmaryrd}
\usepackage{cases}

\makeatother

\providecommand\theoremname{Theorem}
\theoremstyle{plain}
\newtheorem{thm}{\protect\theoremname}
\providecommand\definitionname{Definition}
\theoremstyle{definition}
\newtheorem{defn}[thm]{\protect\definitionname}
\providecommand\lemmaname{Lemma}
\theoremstyle{plain}
\newtheorem{lem}[thm]{\protect\lemmaname}
\providecommand\examplename{Example}
\theoremstyle{definition}
\newtheorem{example}[thm]{\protect\examplename}
\providecommand\corollaryname{Corollary}
\theoremstyle{plain}
\newtheorem{cor}[thm]{\protect\corollaryname}

\begin{document}
\title{Quantum Sufficiency for Self-Adjoint Statistical Models via Likelihood-Type
Operators on Real $*$-Subalgebras and Real Jordan Algebras}
\author{
	Koichi Yamagata%
	\thanks{yamagata@se.kanazawa-u.ac.jp}\\
	{Institute of Science and Engineering, Kanazawa University} \\
	{Kanazawa, Ishikawa, 920-1192, Japan}
}%
\maketitle

\global\long\def\E{\mathcal{E}}%
\global\long\def\S{\mathcal{S}}%
\global\long\def\R{\mathbb{R}}%
\global\long\def\C{\mathbb{C}}%
\global\long\def\N{\mathbb{N}}%
\global\long\def\Z{\mathbb{Z}}%
\global\long\def\D{\mathcal{D}}%
\global\long\def\M{\mathcal{M}}%
\global\long\def\X{\mathcal{X}}%
\global\long\def\F{\mathcal{F}}%
\global\long\def\B{\mathcal{B}}%
\global\long\def\T{\mathcal{T}}%
\global\long\def\P{\mathcal{P}}%
\global\long\def\H{\mathcal{H}}%
\global\long\def\Y{\mathcal{Y}}%
\global\long\def\A{\mathcal{A}}%
\global\long\def\bra#1{\left\langle #1\right|}%
\global\long\def\ket#1{\left|#1\right\rangle }%
\global\long\def\tr{{\rm tr}\,}%
\global\long\def\Tr{{\rm Tr}\,}%
\global\long\def\braket#1#2{\left\langle #1\mid#2\right\rangle }%
\global\long\def\V{\mathcal{V}}%
\global\long\def\re{{\rm Re}\,}%
\global\long\def\im{{\rm Im}\,}%
\global\long\def\id{{\rm id}}%
\global\long\def\ii{\sqrt{-1}}%
\global\long\def\supp{\mathrm{supp}\,}%
\global\long\def\Span{\mathrm{span}}%

\begin{abstract}
We develop a likelihood-operator formulation of quantum sufficiency.
In Petz's theory, sufficiency is characterized in terms of Radon--Nikodym
cocycles and modular structure associated with faithful reference
states. The present paper isolates the likelihood-ratio content of
this theory and reformulates it by means of self-adjoint likelihood-type
operators, such as square-root likelihood ratios and symmetric logarithmic
derivatives. This shift is particularly relevant for degenerate models.
For example, distinct nonorthogonal pure states are not absolutely
continuous with respect to one another in the usual support-theoretic
sense, whereas they are naturally related by a square-root likelihood
ratio. 

The framework is formulated for statistical models consisting of self-adjoint
operators, and therefore includes ordinary state models, local models
given by derivatives of states, and degenerate quantum exponential
families. We introduce sufficiency for real-linear unital positive
maps and study the associated sufficient complex $*$-subalgebras,
real $*$-subalgebras, and real Jordan algebras. From sufficient real
unital positive maps, we construct sufficient Jordan conditional
expectations and sufficient real and complex conditional expectations
onto the corresponding generated algebras. The latter maps are
respectively real and complex completely positive. 

For a reference state $\rho$ and a likelihood-ratio set $\mathcal{R}$,
we characterize the minimal sufficient real $*$-subalgebra as the
smallest real $*$-subalgebra containing $\mathcal{R}$ and invariant
under the $\rho$-modular transformation. We further show that the
minimal sufficient real Jordan algebra is generated by $\mathcal{R}$
together with the operator obtained by normalizing every irreducible
diagonal block of $\rho$, relative to the minimal sufficient real
$*$-subalgebra, to trace one. As structural consequences, we obtain
Koashi--Imoto type decompositions for the resulting real $*$-subalgebras
and real Jordan algebras associated with self-adjoint models.
The resulting theory separates the likelihood-ratio aspect of quantum
sufficiency from its genuinely noncommutative modular aspect, and
provides a unified algebraic treatment of ordinary state models, tangent
models arising from derivatives of states, pure-state models, degenerate
reference states, and rank-degenerate quantum exponential families.
\end{abstract}

\section{Introduction}

Sufficiency is a structural notion in statistical inference. In classical
statistics, a sufficient statistic retains all information about the
unknown parameter contained in the original observation, and the Neyman--Fisher
factorization theorem expresses this property in terms of likelihood
ratios. In quantum statistics, the corresponding theory was initiated
by Petz in the operator-algebraic setting. Petz formulated a noncommutative
analogue of sufficient statistics and showed that sufficiency of a
subalgebra is governed not merely by preservation of a scalar information
quantity, such as relative entropy, but by modular operator-algebraic
structure. In particular, sufficiency can be characterized by the
inclusion of Radon--Nikodym cocycles in the reduced algebra and by
the existence of a suitable generalized conditional expectation \cite{sufficient_petz}.

This cocycle formulation is one of the central achievements of Petz's
theory. It identifies the modular structure behind statistical sufficiency
and has influenced quantum information theory, quantum statistical
inference, and the theory of recovery maps. In finite dimensions,
the same viewpoint extends from subalgebra inclusions to coarse-grainings
between matrix algebras. One asks whether a coarse-graining preserves
all statistical information carried by a family of density operators.
This formulation was developed by Jen\v{c}ov\'a and Petz, who established
several equivalent characterizations of sufficient coarse-grainings
and connected sufficiency with reversibility, factorization-type descriptions,
and modular conditions expressed directly in terms of density operators
\cite{KI-deco2}. Mosonyi and Petz further clarified the internal
structure imposed by sufficient coarse-grainings, showing that sufficient
reductions force a block decomposition of the underlying states \cite{sufficient_moso}.
This line of work is closely related to the Koashi-{}-Imoto decomposition,
which describes the structure of operations preserving a family of
quantum states \cite{KI-deco1,KI-deco2}.

The standard finite-dimensional theory is thus built around four intertwined
ingredients: families of density operators, complex completely positive
coarse-grainings, faithful reference states, and Radon--Nikodym cocycles.
This is a natural framework from the viewpoint of noncommutative integration
and modular theory. The aim of the present paper is not to replace
this modular viewpoint, but to separate within it two different aspects
of quantum sufficiency: the likelihood-ratio aspect and the genuinely
noncommutative modular aspect. The former is the part that has a direct
analogue in the classical Neyman–Fisher factorization theorem, or
equivalently in the fact that classical sufficiency is characterized
by the preservation of likelihood ratios. The latter is the additional
structure forced by noncommutativity.

The distinction becomes visible when one compares cocycles with the
likelihood objects used in quantum estimation. Symmetric logarithmic
derivatives are basic objects in local quantum estimation, while square-root
likelihood ratios were introduced and developed in the q-LAN framework
\cite{qLAN1,qLAN2,qLAN3}. These operators serve as noncommutative
analogues of classical likelihood ratios, but they are self-adjoint
operators rather than cocycles. They live naturally in the real self-adjoint
part of the operator algebra and remain meaningful in situations where
the reference state is not faithful.

A simple example is provided by pure states. In the usual support-theoretic
sense, a state $\sigma$ is absolutely continuous with respect to
a state $\rho$ only if
\[
\supp\sigma\leq\supp\rho.
\]
Consequently, two distinct pure states fail to be absolutely continuous
with respect to one another unless they have the same support, even
when they are nonorthogonal. Thus a smooth pure-state model $\rho_{\theta}=|\psi_{\theta}\rangle\langle\psi_{\theta}|$
is not naturally absolutely continuous with respect to one of its
members in this support-theoretic sense, even locally. By contrast,
the square-root likelihood-ratio relation asks whether there exists
a positive operator $R\geq0$ such that
\[
\sigma=R\rho R.
\]
When such an operator exists, we call any such $R$ a square-root
likelihood ratio of $\sigma$ with respect to $\rho$ \cite{qLAN1,qLAN2,qLAN3}.
This relation can hold for distinct nonorthogonal pure states. Indeed,
if $\rho=|\psi\rangle\langle\psi|$ and $\sigma=|\phi\rangle\langle\phi|$
with $\langle\psi|\phi\rangle\neq0$, then
\[
R=\frac{1}{|\langle\psi|\phi\rangle|}\,|\phi\rangle\langle\phi|
\]
satisfies $\sigma=R\rho R$. Thus square-root likelihood ratios capture
a statistical relation between nonorthogonal pure states that is invisible
to the support-theoretic notion of absolute continuity.

A related point appears in local quantum statistics. For a smooth
parametric family of states $\{\rho_{\theta}\}$, the local model
at $\theta_{0}$ is determined not only by the state $\rho_{\theta_{0}}$
but also by the derivatives $\partial_{i}\rho_{\theta_{0}}$. These
derivatives are self-adjoint operators, but they need not be positive.
The corresponding symmetric logarithmic derivatives (SLDs) $L_{i}$
are defined by
\[
\partial_{i}\rho_{\theta_{0}}=\frac{1}{2}(L_{i}\rho_{\theta_{0}}+\rho_{\theta_{0}}L_{i}),
\]
and they enter the definitions of quantum Fisher information and the
Holevo Cram\'er--Rao bound. Thus, from the viewpoint of quantum
statistical inference, it is natural to formulate models not only
as families of states, but more generally as families of self-adjoint
operators.

This paper develops such a formulation of quantum sufficiency. The
likelihood-ratio aspect is represented by square-root likelihood ratios,
SLDs, and related self-adjoint likelihood-type operators. The modular
aspect is represented by invariance under the transformation
\[
A\longmapsto\rho A\rho^{-1},
\]
where $\rho^{-1}$ is understood as a generalized inverse when $\rho$
is not faithful. This separation makes it possible to treat degenerate
reference states, pure-state models, local models, and degenerate
quantum exponential families in a unified finite-dimensional framework.

There is also an algebraic reason for passing to the self-adjoint
real setting. Observables are self-adjoint operators, their linear
structure is real, and their canonical algebraic closure is Jordan
rather than associative. Moreover, if the purpose is to compare statistical
models or to map POVMs to POVMs, complex complete positivity is not
the minimal statistical requirement. A real-linear unital positive
map already maps POVMs to POVMs: if $\{M_{x}\}_{x\in\mathsf{X}}$
is a POVM, with $M_{x}\geq0$ and $\sum_{x}M_{x}=I$, then $\{\alpha(M_{x})\}_{x\in\mathsf{X}}$
is again a POVM for every real unital positive map $\alpha$. Complete
positivity remains essential when one wants to represent physically
admissible quantum operations, but the statistical comparison of self-adjoint
models naturally leads to the weaker real-positive setting.

Accordingly, the present paper studies quantum sufficiency on three
levels: complex $*$-subalgebras, real $*$-subalgebras, and real
Jordan algebras. Real $*$-subalgebras are real associative $*$-subalgebras
of matrix algebras, and their finite-dimensional structure contains
real, complex, and quaternionic matrix components. Real Jordan algebras
are real subspaces of the self-adjoint part closed under the Jordan
product $A\circ B=(AB+BA)/2$; their finite-dimensional structure
also includes spin factors. These structures appear naturally once
sufficiency is formulated in terms of self-adjoint likelihood-type
operators and real-linear positive maps.

Closely related work has recently been developed by van Luijk and
Wilming \cite{vanluijk_wilming}. They study sufficiency and interconversion
of quantum statistical experiments consisting of states under positive
trace-preserving maps, from the viewpoint of statistical comparison
and Neyman--Pearson testing, and characterize the associated minimal
sufficient Jordan algebras. The present paper takes a different route.
Rather than treating sufficiency primarily as a comparison problem
for state experiments, we reformulate Petz-type sufficiency through
self-adjoint likelihood-type operators and $\rho$-modular invariance.
This framework applies to self-adjoint statistical models, including
derivatives of states, and leads to sufficient real $*$-subalgebras,
sufficient real Jordan algebras, and a formulation of degenerate quantum
exponential families via square-root likelihood ratios.

We now summarize the main results. First, we generalize the notion
of a statistical model from a family of states to a family of self-adjoint
operators on a Hilbert space $\mathcal{H}$. This includes ordinary
state models, local models formed by derivatives of states, and mixtures
of these. Square-root likelihood ratios and symmetric logarithmic
derivatives then arise as the likelihood-ratio data of the model.
The precise definitions of self-adjoint models, absolutely continuous
models, local models, and likelihood-ratio sets are given in Section
\ref{sec:main}.

Second, we introduce sufficient real unital positive maps. A real
unital positive map $\alpha$ is sufficient for a self-adjoint model
$\mathcal{S}$ if
\[
\re\Tr XB=\re\Tr X\alpha(B)\qquad(X\in\mathcal{S})
\]
for all observables $B$. This condition expresses that the statistical
information carried by the model is preserved under the reduction
implemented by $\alpha$. Since real unital positive maps transform
POVMs to POVMs, this gives a statistical notion of reduction for self-adjoint
models without imposing complex complete positivity unless a physical
quantum operation is required.

Third, we establish a hierarchy of sufficient algebraic structures.
The Hermitian fixed points of a sufficient real unital positive map
form a Jordan algebra admitting a sufficient Jordan conditional
expectation. Its generated real and complex $*$-algebras admit
sufficient real and complex conditional expectations, respectively,
which are real and complex completely positive. Thus the usual complex
$*$-algebraic formulation appears as one level of a broader real
and Jordan-algebraic sufficiency theory.

Fourth, we characterize minimal sufficient real $*$-subalgebras.
Given a reference state $\rho$ and a likelihood-ratio set $\mathcal{R}$,
the minimal sufficient real $*$-subalgebra is the smallest real $*$-subalgebra
that contains $\mathcal{R}$ and is invariant under the $\rho$-modular
transformation
\[
A\mapsto\rho A\rho^{-1}.
\]
The analogous statement holds for minimal sufficient complex $*$-subalgebras.
This characterization replaces the cocycle condition by a condition
formulated directly in terms of self-adjoint likelihood-type operators
and modular invariance.

Fifth, we characterize the minimal sufficient real Jordan algebra.
Let $\mathcal{A}$ be the minimal sufficient real $*$-subalgebra,
and let $P_{\lambda}$ project onto its irreducible invariant subspaces,
including every repeated copy. Define
\[
\rho_1=\sum_{\lambda}
\frac{P_{\lambda}\rho P_{\lambda}}
{\Tr(P_{\lambda}\rho P_{\lambda})}.
\]
Then the unital real Jordan algebra generated by
$\mathcal{R}\cup\{\rho_1\}$ is the minimal sufficient real Jordan algebra. Thus the Jordan structure
is not an auxiliary object; it is the natural algebra generated by
the self-adjoint likelihood-ratio data of the model. This also shows
that the likelihood-ratio aspect of sufficiency can be separated from
its genuinely noncommutative modular aspect.

Sixth, we obtain Koashi--Imoto type decompositions in the real $*$-algebraic
and real Jordan settings. Unlike the usual state-based Koashi--Imoto
decomposition, the decompositions obtained here apply to general self-adjoint
models. In the real Jordan case, the decomposition contains not only
complex, real, and quaternionic Hermitian matrix components, but also
spin-factor components. This gives a structural description of sufficient
real Jordan algebras and clarifies how model-dependent and model-independent
parts of a self-adjoint model separate under sufficient reduction.

Finally, we discuss examples illustrating the framework. Pure-state
models show the limitation of the support-theoretic notion of absolute
continuity and the role of square-root likelihood ratios. We also
analyze one- and two-parameter pure-state local models and show that
their minimal sufficient real and Jordan algebras persist for the
corresponding local i.i.d. models on their support spaces. In addition,
we formulate degenerate quantum exponential families using square-root
likelihood ratios, showing that the present
framework naturally includes rank-degenerate cases that are difficult
to express in the traditional cocycle-based language.

The remainder of this paper is organized as follows. Section \ref{sec:main}
introduces the basic framework and summarizes the main results. Section
\ref{sec:real_algebra} reviews the structure of real $*$-subalgebras
and establishes basic facts on real positive maps and real conditional
expectations. Section \ref{sec:modular} is devoted to degenerate
$\rho$-modular invariance and its characterization in terms of real
conditional expectations. Section \ref{sec:jordan_pos} develops the
structure of real Jordan algebras and real positive maps. In Section
\ref{sec:pos2ce}, we show how sufficient real positive maps give
rise to sufficient real Jordan algebras, sufficient real $*$-subalgebras,
and sufficient complex $*$-subalgebras. Section \ref{sec:min_suff}
characterizes minimal sufficient real $*$-subalgebras and minimal
sufficient real Jordan algebras in terms of the likelihood-ratio set
and the $\rho$-modular invariance condition. In Section \ref{sec:KI},
we establish Koashi--Imoto type decompositions in the real $*$-algebraic
and real Jordan settings. Section \ref{sec:support} discusses an
application to sufficient support size in quantum estimation. Section
\ref{sec:examples} presents several examples illustrating the present
framework. Finally, the appendices contain supplementary discussions,
including the necessity of restricting to $\mathcal{H}_{\mathcal{S}}$
and the relation between $\mathcal{D}_{\rho}$-invariance and $\rho$-modular
invariance.

\section{Framework and main results\label{sec:main}}

Let $\mathcal{H}$ be a finite-dimensional Hilbert space. We denote
by $\mathcal{B}(\mathcal{H})$ the algebra of linear operators on
$\mathcal{H}$, and by $\mathcal{B}_{h}(\mathcal{H})$ the set of
Hermitian operators. We write $\mathcal{B}(\mathcal{H})_{+}$ for
the set of positive operators, and $\mathcal{B}(\mathcal{H})_{++}$
for the set of strictly positive operators. We now introduce real
$*$-subalgebras of $\mathcal{B}(\mathcal{H})$.
\begin{defn}
\label{def:real_alg}Let $\A_{\R}$ be a real subspace of $\B(\mathcal{H})$.
We call $\A_{\R}$ a real $*$-subalgebra of $\mathcal{B}(\mathcal{H})$
if the following conditions hold:
\begin{description}
\item [{(i)}] If $A,B\in\A_{\R}$, then $A+B\in\A_{\R}$. 
\item [{(ii)}] If $A,B\in\A_{\R}$, then $AB\in\A_{\R}$.
\item [{(iii)}] If $A\in\A_{\R}$ and $r\in\R$, then $rA\in\A_{\R}$. 
\item [{(iv)}] If $A\in\A_{\R}$, then $A^{*}\in\A_{\R}$.
\item [{(v)}] $I\in\A_{\R}$.
\end{description}
\end{defn}

Note that condition (iii) requires closure only under real scalars
$r\in\mathbb{R}$, not necessarily under complex scalars. In the decomposition
into irreducible components of a finite-dimensional real $*$-subalgebra
of $\mathcal{B}(\mathcal{H})$, one encounters not only complex matrix
algebras $\mathbb{M}_{n}(\mathbb{C})$, but also real matrix algebras
$\mathbb{M}_{n}(\mathbb{R})$ and quaternionic matrix algebras $\mathbb{M}_{n}(\mathbb{H})$.
See Section \ref{sec:real_algebra} for details.

We next introduce real Jordan algebras $\mathcal{A}_{J}$.
\begin{defn}
\label{def:jordan_alg}Let $\mathcal{A}_{J}\subset\mathcal{B}_{h}(\mathcal{H})$.
We call $\mathcal{A}_{J}$ a real Jordan algebra if the following
conditions hold:
\begin{description}
\item [{(i)}] If $A,B\in\mathcal{A}_{J}$, then $A+B\in\mathcal{A}_{J}$. 
\item [{(ii)}] If $A,B\in\mathcal{A}_{J}$, then $A\circ B\in\mathcal{A}_{J}$.
\item [{(iii)}] If $A\in\mathcal{A}_{J}$ and $r\in\R$, then $rA\in\mathcal{A}_{J}$. 
\item [{(iv)}] $I\in\mathcal{A}_{J}$.
\end{description}
Here $A\circ B=\frac{1}{2}\left(AB+BA\right)$ denotes the Jordan
product. In the decomposition into irreducible components of a finite-dimensional
real Jordan algebra, one encounters not only the Hermitian parts of
$\mathbb{M}_{n}(\mathbb{C})$, $\mathbb{M}_{n}(\mathbb{R})$, and
$\mathbb{M}_{n}(\mathbb{H})$, but also spin factors $\Gamma_{n}$.
See Section \ref{sec:jordan_pos} for details.
\end{defn}

We now introduce several classes of real positive maps. A real-linear
map $\alpha:\mathcal{B}(\mathcal{H})\to\mathcal{B}(\mathcal{H})$
is called positive if $\alpha(B)\ge0$ for every $B\in\mathcal{B}(\mathcal{H})_{+}$.
A real positive map $\alpha$ is called unital if $\alpha(I)=I$.
A real positive map $\alpha$ is called faithful if, for every $B\in\mathcal{B}(\mathcal{H})_{+}$,
the equality $\alpha(B)=0$ implies $B=0$. Let $\mathcal{A}_{\mathbb{R}}\subset\mathcal{B}(\mathcal{H})$
be a real $*$-subalgebra. A real unital positive map $\alpha:\mathcal{B}(\mathcal{H})\to\mathcal{A}_{\mathbb{R}}$
is called a real conditional expectation if
\[
\alpha(B^*)=\alpha(B)^*,\qquad A\alpha(B)=\alpha(AB)
\]
holds for every $A\in\mathcal{A}_{\mathbb{R}}$ and $B\in\mathcal{B}(\mathcal{H})$.
These identities also give $\alpha(BA)=\alpha(B)A$.
We say that a real $*$-preserving map $\alpha$ is real completely
positive if, for every finite
collections $\{A_{i}\}^{n}_{i=1}\subset\mathcal{A}_{\mathbb{R}},\{B_{i}\}^{n}_{i=1}\subset\mathcal{B}(\mathcal{H}),$
one has
\[
\sum_{i,j}A^{*}_{i}\alpha(B^{*}_{i}B_{j})A_{j}\geq0.
\]
Let $\mathcal{A}_{J}\subset\mathcal{B}_{h}(\mathcal{H})$ be a real
Jordan algebra. A real unital positive map $\alpha:\mathcal{B}_{h}(\mathcal{H})\to\mathcal{A}_{J}$
is called a real Jordan conditional expectation if
\[
A\circ\alpha(B)=\alpha(A\circ B)
\]
holds for every $A\in\mathcal{A}_{J}$ and $B\in\mathcal{B}_{h}(\mathcal{H})$.
Note that all maps introduced here are assumed to be real-linear,
not complex-linear. When one wishes to describe physically admissible
transitions, one often requires complex complete positivity. By contrast,
when the purpose is to compare statistical models and map POVMs to
POVMs, it is sufficient to work with real-linear positive maps. Indeed,
if $\{M_{x}\}_{x\in\mathsf{X}}$ is a POVM, that is, $M_{x}\in\mathcal{B}(\mathcal{H})_{+}$
and $\sum_{x}M_{x}=I$, then a real unital positive map $\alpha$
produces another POVM $\{\alpha(M_{x})\}_{x\in\mathsf{X}}$. This
distinction is one of the key points of the present work.

We next introduce the notion of a model. Throughout this paper, a
model on a finite-dimensional Hilbert space $\mathcal{H}$ means a
family $\mathcal{S}\subset\mathcal{B}_{h}(\mathcal{H})$ of self-adjoint
operators. Note that, unlike in the conventional setting, a model
in our sense is not restricted to families of states or positive operators.
We next introduce two important types of models: absolutely continuous
models and local models. An absolutely continuous model is defined
as follows. Fix a reference state $\rho\in\mathcal{S}$ with $\rho\ge0$,
and assume that for every $X\in\mathcal{S}$ there exists an operator
$R_{X}\ge0$ such that 
\begin{equation}
X=R_{X}\rho R_{X}.\label{eq:ratio}
\end{equation}
This setting is intended for families of quantum states. In \cite{qLAN2},
when such an operator $R_{X}$ exists, $X\geq0$ is said to be absolutely
continuous with respect to $\rho$, and $R_{X}$ is called the square-root
likelihood ratio. We emphasize, however, that the notion of absolute
continuity used here should not be understood in the conventional
support-theoretic sense. A local model is defined as follows. Fix
a reference state $\rho\in\mathcal{S}$ with $\rho\ge0$, and assume
that every $X\in\mathcal{S}$ can be written in the form
\begin{equation}
X=\frac{1}{2}(L_{X}\rho+\rho L_{X})\label{eq:sld}
\end{equation}
for some self-adjoint operator $L_{X}\in\mathcal{B}_{h}(\mathcal{H})$.
This setting is motivated by a smoothly parameterized family of quantum
states $\left\{ \rho_{\theta}\in\mathcal{B}(\mathcal{H})_{+}\mid\theta\in\Theta\subset\mathbb{R}^{d}\right\} $.
Fixing $\theta_{0}\in\Theta$, we have in mind the model $\mathcal{S}=\{\rho_{\theta_{0}}\}\cup\left\{ \frac{\partial}{\partial\theta^{i}}\rho_{\theta_{0}}\right\} ^{d}_{i=1}$.
Taking $\rho=\rho_{\theta_{0}}$ as the reference state, let $L_{i}\in\mathcal{B}_{h}(\mathcal{H})$
be defined by $\frac{\partial}{\partial\theta^{i}}\rho_{\theta_{0}}=\frac{1}{2}(L_{i}\rho+\rho L_{i})$.
We call $L_{i}$ the symmetric logarithmic derivative (SLD) associated
with $\frac{\partial}{\partial\theta^{i}}\rho_{\theta_{0}}$. In general,
elements of a local model $\mathcal{S}$ need not be positive operators.
One may also consider models in which absolutely continuous and local
components are mixed.

For a model $\mathcal{S}$, define a subspace $\H_{\S}\subset\mathcal{H}$
by
\[
\H_{\S}:=\Span_{\C}\{X(\mathcal{H})\mid X\in\mathcal{S}\}.
\]
Then every $X\in\mathcal{S}$ may be regarded as an operator on $\H_{\S}$.
Since for every $B\in\mathcal{B}(\mathcal{H})$ and $X\in\mathcal{S}$,
\[
\Tr XB=\Tr X(PBP^{*})
\]
where $P:\mathcal{H}\to\H_{\S}$ denotes the orthogonal projection,
there is no loss in restricting all operators to $\mathcal{B}(\H_{\S})$.
If one does not restrict to $\H_{\S}$, pathological examples may
occur; see Appendix \ref{sec:ill_ex}. This is true not only in the
present setting but also in the conventional theory and even in classical
statistics. In the conventional approach, each $X\in\mathcal{S}$
is assumed to satisfy $X\ge0$, and one chooses $\sigma\in\Span_{\R}\mathcal{S}$
such that for every $X\in\mathcal{S}$ there exists $t>0$ with $t\sigma\ge X$.
The Hilbert space is then restricted to $(\ker\sigma)^{\perp}\subset\mathcal{H}$.
In the present work, no such operator $\sigma$ is needed. 

For a model $\S$, a real unital positive map $\alpha:\mathcal{B}(\H_{\S})\to\mathcal{B}(\H_{\S})$
is called sufficient if, for every $B\in\mathcal{B}(\H_{\S})$,
\begin{equation}
\re\Tr XB=\re\Tr X\alpha(B)\qquad\forall X\in\S.\label{eq:sufficient}
\end{equation}
A real $*$-subalgebra $\mathcal{A}_{\mathbb{R}}\subset\mathcal{B}(\H_{\S})$
is called a sufficient real $*$-subalgebra if there exists a sufficient
real unital positive map $\alpha$ such that $\alpha(\B(\H_{\S}))\subset\A_{\R}$.
A real Jordan algebra $\mathcal{A}_{J}\subset\mathcal{B}_{h}(\H_{\S})$
is called sufficient if there exists a sufficient real unital positive
map $\alpha:\mathcal{B}(\H_{\S})\to\mathcal{B}(\H_{\S})$ such that
$\alpha\bigl(\mathcal{B}_{h}(\H_{\S})\bigr)\subset\mathcal{A}_{J}$.
We equip $\mathcal{B}(\H_{\S})$ with the real Hilbert--Schmidt inner
product
\[
\left\langle A,B\right\rangle :=\re\Tr A^{*}B\qquad A,B\in\mathcal{B}(\H_{\S}).
\]

We prove the following theorem.
\begin{thm}
\label{thm:pos_to_CE}Let $\mathcal{S}$ be a model, and let $\alpha:\mathcal{B}(\H_{\S})\to\mathcal{B}(\H_{\S})$
be a sufficient real unital positive map. Let $\mathcal{A}_{J}\subset\mathcal{B}_{h}(\H_{\S})$
be the set of fixed points of $\alpha$ in $\mathcal{B}_{h}(\H_{\S})$.
Then $\mathcal{A}_{J}$ is a real Jordan algebra. Let $\mathcal{A}_{\mathbb{R}}$
and $\mathcal{A}_{\mathbb{C}}$ be the real and complex $*$-subalgebras
generated by $\mathcal{A}_{J}$, respectively. Then there exist a
sufficient real Jordan conditional expectation $\beta_{J}:\mathcal{B}_{h}(\H_{\S})\to\mathcal{A}_{J}$,
a sufficient real conditional expectation $\beta_{\mathbb{R}}:\mathcal{B}(\H_{\S})\to\mathcal{A}_{\mathbb{R}}$,
and a sufficient complex conditional expectation $\beta_{\mathbb{C}}:\mathcal{B}(\H_{\S})\to\mathcal{A}_{\mathbb{C}}$,
and these maps can be chosen to be faithful.

Moreover, there exists $\omega\in\mathcal{B}(\H_{\S})_{++}$ such
that $\beta_{J}$, $\beta_{\mathbb{R}}$, and $\beta_{\mathbb{C}}$
are sufficient for $\mathcal{S}\cup\{\omega\}$, and satisfy
\[
[\omega,A]=0\qquad\forall A\in\A_{\C},
\]
\[
[\omega,X]=0\quad\forall X\in\S,
\]
\[
\Pi_{J}(\omega)=I,
\]
and
\begin{equation}
\Pi_{J}(X)=X\omega^{-1}\quad\forall X\in\S.\label{eq:sigma-omega}
\end{equation}
For $\mathbb{K}=\mathbb{R},\mathbb{C}$, the operators
$z_{\mathbb{K}}:=\Pi_{\mathbb{K}}(\omega)$ are strictly positive
and central in $\mathcal{A}_{\mathbb{K}}$, and
\[
\Pi_{\mathbb{K}}(X)=z_{\mathbb{K}}\Pi_{J}(X)
\quad\forall X\in\mathcal{S}.
\]
Here $\Pi_{J}$, $\Pi_{\mathbb{R}}$, and $\Pi_{\mathbb{C}}$ denote
the orthogonal projections from $\mathcal{B}(\H_{\S})$ onto $\mathcal{A}_{J}$,
$\mathcal{A}_{\mathbb{R}}$, and $\mathcal{A}_{\mathbb{C}}$, respectively,
with respect to the real Hilbert--Schmidt inner product.
\end{thm}

We call $\omega$ a supporting operator. With the real inner product
on $\mathcal{B}(\H_{\S})$ defined by $\left\langle A,B\right\rangle _{\omega}=\re\Tr\omega A^{*}B$
for $A,B\in\mathcal{B}(\H_{\S})$, the maps $\beta_{J}$, $\beta_{\mathbb{R}}$,
and $\beta_{\mathbb{C}}$ become the orthogonal projections with respect
to this inner product. The supporting operator $\omega$ does not
necessarily belong to $\Span_{\R}\S$. As an extreme example, $\mathcal{B}(\H_{\S})$
itself may be a sufficient algebra; in this case $\omega=I$, and
in general $I\notin\Span_{\R}\S$.

We next introduce the likelihood-ratio set, which will be used to
characterize minimal sufficient $*$-subalgebras. For an absolutely
continuous model $\mathcal{S}$ with reference state $\rho\in\mathcal{S}$,
the set $\mathcal{R}=\{R_{X}\mid X\in\mathcal{S}\}$ of square-root
likelihood ratios is called the likelihood-ratio set. Likewise, for
a local model $\mathcal{S}$ with reference state $\rho\in\mathcal{S}$,
the set $\mathcal{R}=\{L_{X}\mid X\in\mathcal{S}\}$ is called the
likelihood-ratio set. When $\rho$ is not strictly positive, the operators
$R_{X}$ and $L_{X}$ are in general not uniquely determined by $X\in\mathcal{S}$.
In such cases, we always choose the one with minimal Hilbert--Schmidt
norm. We are now in a position to characterize minimal sufficient
real and complex $*$-subalgebras in terms of the likelihood-ratio
set together with the $\rho$-invariance condition. 
\begin{thm}
\label{thm:min_suff_alg} Let $\mathcal{S}$ be a model with reference
state $\rho\in\mathcal{S}$ and likelihood-ratio set $\mathcal{R}$.
Here and below, $\rho^{-1}$ denotes the generalized inverse of $\rho$.
Let $\mathcal{A}\subset\mathcal{B}(\H_{\S})$ be a real $*$-subalgebra
satisfying
\begin{description}
\item [{(i)}] $\rho A\rho^{-1}\in\A$ for every $A\in\mathcal{A}$,
\item [{(ii)}] $\mathcal{R}\subset\A$ . 
\end{description}
Then $\mathcal{A}$ is a sufficient real $*$-subalgebra admitting
a sufficient real conditional expectation. If $\mathcal{A}$ is moreover
a complex $*$-subalgebra, then it admits a sufficient complex conditional
expectation. 

Furthermore, the smallest real $*$-subalgebra satisfying (i) and
(ii) is the minimal sufficient real $*$-subalgebra, and the smallest
complex $*$-subalgebra satisfying (i) and (ii) is the minimal sufficient
complex $*$-subalgebra.
\end{thm}

Note that $\rho$ need not be strictly positive. In the work \cite{qLAN3},
which established the asymptotic representation theorem in q-LAN,
the $\D$-invariant extension of square-root likelihood ratios played
a crucial role. Theorem \ref{thm:min_suff_alg} is closely related
to that theory. Indeed, condition (i) has a natural interpretation
in terms of $\D_{\rho}$-invariant extension; see Appendix \ref{sec:dinv}.
We next characterize the minimal sufficient real Jordan algebra.
\begin{thm}
\label{thm:min_suff_J}In the setting of Theorem \ref{thm:min_suff_alg},
let $\mathcal{A}$ be the minimal sufficient real $*$-subalgebra.
Choose an orthogonal decomposition
\[
\mathcal{H}_{\mathcal{S}}=\bigoplus_{\lambda}\mathcal{H}_{\lambda}
\]
into irreducible complex $\mathcal{A}$-invariant subspaces, retaining
each repeated copy, and let $P_{\lambda}$ be the corresponding
orthogonal projections. Then $\Tr(P_{\lambda}\rho P_{\lambda})>0$, and
\[
\rho_1:=\sum_{\lambda}
\frac{P_{\lambda}\rho P_{\lambda}}
{\Tr(P_{\lambda}\rho P_{\lambda})}
\in\mathcal{A}_h.
\]
The unital real Jordan algebra $\mathcal{A}_{J}$ generated by
$\mathcal{R}\cup\{\rho_1\}$ admits a sufficient real Jordan
conditional expectation and is contained in every sufficient real
Jordan algebra for $\mathcal{S}$.
In particular, it is the minimal sufficient real Jordan algebra.
\end{thm}

This theorem indicates that self-adjoint likelihood-type objects,
rather than Radon–Nikodym cocycles, should be regarded as the appropriate
noncommutative analogue of classical likelihood ratios.

We next establish a real Jordan algebra version of the Koashi--Imoto
decomposition. Here, as throughout the paper, the underlying model
is not assumed to consist of states, but may be a family of general
self-adjoint operators. Every finite-dimensional real Jordan algebra
contained in $\mathcal{B}_{h}(\mathcal{H})$ is abstractly isomorphic
to a direct sum of the form \cite{jordan1,jordan2}
\begin{equation}
\A_{J}\cong\left\{ \bigoplus_{i}\mathbb{M}_{n_{i}}(\C)_{h}\right\} \oplus\left\{ \bigoplus_{j}\mathbb{M}_{n_{j}}(\R)_{h}\right\} \oplus\left\{ \bigoplus_{k}\mathbb{M}_{n_{k}}(\mathbb{H})_{h}\right\} \oplus\left\{ \bigoplus_{l}\Gamma_{n_{l}}\right\} .\label{eq:jordan_deco}
\end{equation}
This is an abstract Jordan isomorphism, not a unitary equivalence
of the given representations. A single simple factor may act through
several inequivalent irreducible representations.
Here, $\mathbb{M}_{n}(\mathbb{C})$ denotes the algebra of $n\times n$
complex matrices, $\mathbb{M}_{n}(\mathbb{R})$ the algebra of $n\times n$
real matrices, and $\mathbb{M}_{n}(\mathbb{H})$ the algebra of $n\times n$
quaternionic matrices. The algebra $\mathbb{M}_{n}(\mathbb{H})$ can
be represented by $(2n)\times(2n)$ complex matrices (see Section
\ref{sec:real_algebra}). The spaces $\mathbb{M}_{n}(\mathbb{C})_{h}$, $\mathbb{M}_{n}(\mathbb{R})_{h}$,
and $\mathbb{M}_{n}(\mathbb{H})_{h}$ denote their Hermitian parts.
The factor $\Gamma_{n}$ denotes a spin factor: there exist Hermitian
operators $\{\gamma_{i}\}^{n}_{i=1}$ satisfying
$\gamma_{i}\circ\gamma_{j}=\delta_{ij}I$, and
$\Gamma_{n}=\Span_{\mathbb{R}}\bigl(\{\gamma_{i}\}^{n}_{i=1}\cup\{I\}\bigr).$
See Sections \ref{sec:real_algebra} and \ref{sec:jordan_pos}
for details. The following theorem describes the given representation
of a sufficient real Jordan algebra.
\begin{thm}
\label{thm:KI_deco}Suppose that $\mathcal{A}_{J}\subset\mathcal{B}_{h}(\mathcal{H}_{\mathcal{S}})$
admits a sufficient real Jordan conditional expectation.
Fix a Jordan isomorphism
$\theta:\mathcal{A}_{J}\to\bigoplus_{\alpha}J_{\alpha}$ into simple
factors as in (\ref{eq:jordan_deco}).
There exist a unitary $U$ and, for each $\alpha$, a nonempty family
of pairwise unitarily inequivalent irreducible unital Jordan representations
$\pi_{\alpha,s}:J_{\alpha}\to\mathbb{M}_{d_{\alpha,s}}(\mathbb{C})_h$
$(s\in\Lambda_{\alpha})$, with multiplicities $m_{\alpha,s}\ge1$,
such that
\begin{align*}
UaU^*&=\bigoplus_{\alpha}\bigoplus_{s\in\Lambda_{\alpha}}
\pi_{\alpha,s}\bigl(\theta(a)_{\alpha}\bigr)\otimes I_{m_{\alpha,s}},
&&a\in\mathcal{A}_{J},\\
UXU^*&=\bigoplus_{\alpha}\bigoplus_{s\in\Lambda_{\alpha}}
\pi_{\alpha,s}\bigl(x_\alpha(X)\bigr)\otimes P_{\alpha,s},
&&X\in\mathcal{S},
\end{align*}
where $x_\alpha(X)\in J_\alpha$ and the $P_{\alpha,s}$ are real
diagonal $m_{\alpha,s}\times m_{\alpha,s}$ matrices with strictly
positive eigenvalues, independent of $X$.
For each $\alpha$, all representation blocks share the same
coefficient $x_\alpha(X)$. The representation types are described
in Section \ref{sec:KI}.
\end{thm}

This theorem extends the Koashi--Imoto decomposition to the setting
of real Jordan algebras for self-adjoint models. For many statistical
problems concerning a statistical model $\mathcal{S}$, one can work
entirely within the framework of real Jordan algebras. Indeed, in
Ref. \cite{support}, the dimension $\dim\mathcal{A}_{J}$ of a sufficient
real Jordan algebra plays an important role in estimating the size
of sufficient POVMs in quantum estimation problems.

Theorem \ref{thm:pos_to_CE} is proved in Section \ref{sec:pos2ce},
Theorems \ref{thm:min_suff_alg} and \ref{thm:min_suff_J} are proved
in Section \ref{sec:min_suff}, and Theorem \ref{thm:KI_deco} is
proved in Section \ref{sec:KI}. 

\section{Real $*$-Subalgebras\label{sec:real_algebra}}

In this section, we discuss in more detail real $*$-subalgebras $\mathcal{A}_{\mathbb{R}}$
as defined in Definition \ref{def:real_alg}. A complex-type simple real
algebra may occur in both its standard and conjugate representations.
Put
\[
\mathcal{B}_{n}^{p,q}:=\{(A\otimes I_p)\oplus(\overline{A}\otimes I_q):
 A\in\mathbb{M}_{n}(\mathbb{C})\},
\]
where $p,q\ge0$, $p+q>0$, and summands of multiplicity zero are omitted.
After a suitable unitary change of basis, the represented real algebra
admits a decomposition of the form \cite{real_alg1,real_alg2}
\begin{equation}
\A_{\R}\cong\left\{\bigoplus_i\mathcal{B}_{n_i}^{p_i,q_i}\right\}
\oplus\left\{\bigoplus_j\mathbb{M}_{n_j}(\mathbb{R})\otimes I_{m_j}\right\}
\oplus\left\{\bigoplus_k\mathbb{M}_{n_k}(\mathbb{H})\otimes I_{m_k}\right\}.
\label{eq:real_decomp}
\end{equation}
Here, $\mathbb{M}_{n}(\mathbb{C})$ denotes the algebra of $n\times n$
complex matrices, $\mathbb{M}_{n}(\mathbb{R})$ the algebra of $n\times n$
real matrices, and $\mathbb{M}_{n}(\mathbb{H})$ the algebra of $n\times n$
quaternionic matrices. The algebra $\mathbb{M}_{n}(\mathbb{H})$ can
be represented by $(2n)\times(2n)$ complex matrices. More precisely,
in the standard complex realization, each quaternionic entry is represented
by a $2\times2$ complex matrix of the form
\[
wI_{2}+\ii(xX+yY+zZ),\qquad w,x,y,z\in\mathbb{R},
\]
where $X,Y,Z$ are the Pauli matrices. Also, $I_m$ denotes the
$m\times m$ identity matrix. If the represented algebra is complex,
only the complex-type factors occur, and one may take $q_i=0$.
The two sectors of $\mathcal{B}_{n_i}^{p_i,q_i}$ share the same matrix
$A$ and are not independent algebra summands.

We next prove the following lemma and theorem concerning positive
maps associated with real $*$-subalgebras.
\begin{lem}
\label{lem:CE2CP}Let $\mathcal{A}\subset\mathcal{B}(\mathcal{H})$
be a real $*$-subalgebra. Then every real conditional expectation
$\alpha:\mathcal{B}(\mathcal{H})\to\mathcal{A}$ is real completely
positive.
\end{lem}

\begin{proof}
Since $\alpha$ is real positive and a real conditional expectation,
for every finite collections $\{A_{i}\}^{n}_{i=1}\subset\mathcal{A},\{B_{i}\}^{n}_{i=1}\subset\mathcal{B}(\mathcal{H}),$
we have
\[
\sum_{i,j}A^{*}_{i}\,\alpha(B^{*}_{i}B_{j})\,A_{j}=\alpha\left(\sum_{i,j}A^{*}_{i}B^{*}_{i}B_{j}A_{j}\right)\ge0.
\]
Hence $\alpha$ is real completely positive.
\end{proof}

\begin{thm}
\label{thm:proj}Let $\mathcal{A}\subset\mathcal{B}(\mathcal{H})$
be a real $*$-subalgebra, and let $\Pi:\mathcal{B}(\mathcal{H})\to\mathcal{A}$
be the orthogonal projection with respect to the real Hilbert--Schmidt
inner product. Then the following hold:
\begin{description}
\item [{(i)}] For every $B\in\mathcal{B}(\mathcal{H})$, $\Pi(B^{*})=\Pi(B)^{*}$\textup{.}
\item [{(ii)}] $\Pi$ is a real conditional expectation.
\item [{(iii)}] $\Pi$ is real completely positive. 
\item [{(iv)}] $\Pi$ is faithful. 
\end{description}
If, moreover, $\mathcal{A}$ is a complex $*$-subalgebra, then $\Pi$
is complex completely positive and a complex conditional expectation.
\end{thm}

\begin{proof}
We first prove (i). It suffices to show that $\re\Tr A\Pi(B^{*})=\re\Tr A\Pi(B)^{*}$
for every $A\in\mathcal{A}$ and $B\in\mathcal{B}(\mathcal{H})$.
Using the real Hilbert--Schmidt inner product $\langle\cdot,\cdot\rangle$,
we compute
\begin{align*}
\re\Tr A\Pi(B^{*}) & =\left\langle A^{*},\Pi(B^{*})\right\rangle =\left\langle A^{*},B^{*}\right\rangle \\
 & =\left\langle B,A\right\rangle =\left\langle \Pi(B),A\right\rangle =\re\Tr A\Pi(B)^{*}
\end{align*}
This proves (i).

Next we prove (ii) and (iii). Let $A\in\mathcal{A}$ and $B\in\mathcal{B}(\mathcal{H})$.
To show that $A\,\Pi(B)=\Pi(AB),$ it suffices to verify that $\re\Tr C^{*}A\Pi(B)=\re\Tr C^{*}\Pi(AB)$
for every $C\in\mathcal{A}$. Indeed, 
\begin{align*}
\re\Tr C^{*}A\Pi(B) & =\left\langle A^{*}C,\Pi(B)\right\rangle =\left\langle A^{*}C,B\right\rangle =\re\Tr C^{*}AB\\
 & =\left\langle C,AB\right\rangle =\left\langle C,\Pi(AB)\right\rangle =\re\Tr C^{*}\Pi(AB)
\end{align*}
Thus $\Pi$ has the required module property.

We next show that $\Pi(B)\ge0$ for every $B\in\mathcal{B}(\mathcal{H})_{+}$.
By (i), $\Pi(B)$ is self-adjoint. Let $\Pi(B)=\sum_{i}\lambda_{i}P_{i}$
be its spectral decomposition, where $\lambda_{i}\in\mathbb{R}$ and
$P_{i}\in\mathcal{A}$ are the spectral projections. Then
\[
\Tr P_{i}\Pi(B)=\re\Tr P_{i}\Pi(B)=\left\langle P_{i},\Pi(B)\right\rangle =\left\langle P_{i},B\right\rangle =\Tr P_{i}B\geq0
\]
Since $\Tr P_{i}\Pi(B)=\lambda_{i}\Tr P_{i}$, it follows that $\lambda_{i}\ge0$
for all $i$, and hence $\Pi(B)\ge0$. Thus $\Pi$ is positive and is a real conditional expectation.
The real complete positivity of $\Pi$ now follows from Lemma \ref{lem:CE2CP}.

Assume now that $\mathcal{A}$ is a complex $*$-subalgebra. To prove
the last statement in this theorem, it suffices to show that $\Pi(\ii B)=\ii\Pi(B)$
for all $B\in\mathcal{B}(\mathcal{H})$. Since $\ii I\in\mathcal{A}$,
this follows from the fact that $\Pi$ is a real conditional expectation.
Hence $\Pi$ is complex-linear, and therefore complex completely positive
and a complex conditional expectation.

Finally, we prove (iv). Let $B\in\mathcal{B}(\mathcal{H})_{+}$ satisfy
$\Pi(B)=0$. Then
\[
\Tr B=\left\langle I,B\right\rangle =\left\langle I,\Pi(B)\right\rangle =\Tr\Pi(B)=0
\]
Since $B\ge0$, this implies $B=0$. Hence $\Pi$ is faithful.
\end{proof}

\begin{example}
When $\mathcal{A}=\mathbb{M}_{n}(\mathbb{R})$, the orthogonal projection
$\Pi$ is real completely positive, but not complex completely positive.
Indeed, in the case $n=2$, 
\[
\Pi\!\left(\begin{pmatrix}0 & 0\\
1 & 0
\end{pmatrix}\right)=\begin{pmatrix}0 & 0\\
1 & 0
\end{pmatrix},\qquad\Pi\!\left(\begin{pmatrix}0 & 0\\
\ii & 0
\end{pmatrix}\right)=\begin{pmatrix}0 & 0\\
0 & 0
\end{pmatrix}.
\]
On the other hand, if we restrict $\Pi$ to self-adjoint operators
and consider $\Pi_{h}:\mathcal{B}_{h}(\mathcal{H})\to\mathcal{A}\cap\mathcal{B}_{h}(\mathcal{H}),$
then
\[
\Pi'(B):=\Pi_{h}\!\left(\frac{B+B^{*}}{2}\right)+\ii\,\Pi_{h}\!\left(\frac{B-B^{*}}{2\ii}\right)=\frac{1}{2}(B+B^{T})
\]
defines a complex-linear map. However, $\Pi'$ is no longer completely
positive, and moreover $\Pi'(\mathcal{B}(\mathcal{H}))\not\subset\mathcal{A}$.
\end{example}

\section{Degenerate $\rho$-Modular Invariance and Real Conditional Expectations\label{sec:modular}}

In this section, we show that, for a positive operator $\rho\in\mathcal{B}(\mathcal{H})_{+}$
on a finite-dimensional Hilbert space $\mathcal{H}$ and a real $*$-subalgebra
$\mathcal{A}\subset\mathcal{B}(\mathcal{H})$, the existence of a
real conditional expectation preserving $\rho$, together with
$\supp\rho\in\mathcal{A}$, is characterized by $\rho$-modular invariance.
Here, $\rho$ is not assumed to be strictly positive.
\begin{thm}
\label{thm:modular} Let $\mathcal{A}\subset\mathcal{B}(\mathcal{H})$
be a real $*$-subalgebra, let $\rho\in\mathcal{B}(\mathcal{H})_{+}$,
and let $\Pi:\mathcal{B}(\mathcal{H})\to\mathcal{A}$ be the orthogonal
projection with respect to the real Hilbert--Schmidt inner product.
Put $\rho_{0}=\Pi(\rho)$, and let $s=\supp\rho$ be the support projection
of $\rho$. Let $\rho^{-1}$ and $\rho^{-1}_{0}$ denote the generalized
inverses. Then the following are equivalent:
\begin{description}
\item [{(i)}] There exists a real conditional expectation $\alpha:\mathcal{B}(\mathcal{H})\to\mathcal{A}$
such that $\re\Tr\rho B=\re\Tr\rho\alpha(B)$ for all $B\in\B(\H)$
and $s\in\mathcal{A}$.
\item [{(ii)}] For every $A\in\mathcal{A}$, 
\[
\rho A\rho^{-1}=\rho_{0}A\rho^{-1}_{0}.
\]
\item [{(iii)}] For every $A\in\mathcal{A}$, 
\[
\rho A\rho^{-1}\in\mathcal{A}.
\]
\end{description}
If, moreover, $\mathcal{A}$ is a complex $*$-subalgebra, then the
map $\alpha$ in (i) can be chosen to be a complex conditional expectation.
\end{thm}

\begin{proof}
(i)$\Rightarrow$(ii). It suffices to show that $\re\Tr\rho(\rho^{-1}A\rho)B=\re\Tr\rho(\rho^{-1}_{0}A\rho_{0})B$
for every $A\in\mathcal{A}$ and $B\in\mathcal{B}(\mathcal{H})$.
By Lemma \ref{lem:supp}, we have $s=\supp\rho=\supp\rho_{0}\in\A$.
Hence
\[
\re\Tr\rho(\rho^{-1}A\rho)B=\re\Tr sA\rho B=\re\Tr\rho\alpha(B)sA.
\]
On the other hand,
\begin{align*}
\re\Tr\rho(\rho^{-1}_{0}A\rho_{0})B & =\re\Tr\rho\alpha((\rho^{-1}_{0}A\rho_{0})B)\\
 & =\re\Tr\rho_{0}\alpha((\rho^{-1}_{0}A\rho_{0})B)\\
 & =\re\Tr\alpha(\rho_{0}(\rho^{-1}_{0}A\rho_{0})B)\\
 & =\re\Tr\alpha(sA\rho_{0}B)\\
 & =\re\Tr\rho\alpha(B)sA.
\end{align*}
Thus (i)$\Rightarrow$(ii).

(ii)$\Rightarrow$(iii). This is immediate.

(iii)$\Rightarrow$(ii). We first show that $s\in\mathcal{A}$. Since
$\rho I\rho^{-1}=s$, condition (iii) implies $s\in\mathcal{A}$.
By Lemma \ref{lem:supp}, we have $s=\supp\rho=\supp\rho_{0}.$ It
therefore suffices to show that $\re\Tr\rho_{0}C\rho A\rho^{-1}=\re\Tr\rho_{0}C\rho_{0}A\rho^{-1}_{0}$
for every $C\in\mathcal{A}$. Indeed,
\[
\re\Tr\rho_{0}C\rho A\rho^{-1}=\re\Tr\rho C\rho A\rho^{-1}=\re\Tr C\rho As=\re\Tr C\rho_{0}As=\re\Tr\rho_{0}C\rho_{0}A\rho^{-1}_{0}.
\]
Thus (iii)$\Rightarrow$(ii).

(ii)$\Rightarrow$(i). We first show that $s\in\mathcal{A}$. Since
$\rho I\rho^{-1}=\rho_{0}I\rho^{-1}_{0}=s$, we indeed have $s\in\mathcal{A}$.

We begin with the case $\rho>0$. Define
\begin{equation}
\alpha(B)=\rho^{-1/2}_{0}\Pi(\rho^{1/2}B\rho^{1/2})\rho^{-1/2}_{0},\label{eq:CE0}
\end{equation}
for $B\in\B(\H)$. Let $\Delta$ and $\Delta_{0}$ be the linear maps
on $\mathcal{B}(\mathcal{H})$ defined by
\[
\Delta(B)=\rho B\rho^{-1},\qquad\Delta_{0}(B)=\rho_{0}B\rho^{-1}_{0}.
\]
By assumption (ii), we have $\Delta(A)=\Delta_{0}(A)$ for $A\in\mathcal{A}$.
Since $\Delta$ and $\Delta_{0}$ are self-adjoint operators on $\mathcal{B}(\mathcal{H})$
with respect to the Hilbert--Schmidt inner product, and their spectra
are positive, there exist numbers $0<\lambda_{\min}\le\lambda_{\max}$
such that the spectra of both $\Delta$ and $\Delta_{0}$ are contained
in $[\lambda_{\min},\lambda_{\max}]$. Hence for every continuous
function $f:[\lambda_{\min},\lambda_{\max}]\to\mathbb{R}$, we have
$f(\Delta)(A)=f(\Delta_{0})(A)$. In particular, $\Delta^{1/2}(A)=\Delta^{1/2}_{0}(A)$,
that is,
\[
\rho^{1/2}A\rho^{-1/2}=\rho^{1/2}_{0}A\rho^{-1/2}_{0}.
\]
We now show that $A\alpha(B)=\alpha(AB)$ for every $A\in\mathcal{A},\ B\in\mathcal{B}(\mathcal{H})$.
Indeed,
\begin{align*}
\alpha(AB) & =\rho^{-1/2}_{0}\Pi(\rho^{1/2}AB\rho^{1/2})\rho^{-1/2}_{0}=\rho^{-1/2}_{0}\Pi(\rho^{1/2}A\rho^{-1/2}\rho^{1/2}B\rho^{1/2})\rho^{-1/2}_{0}\\
 & =\rho^{-1/2}_{0}\rho^{1/2}_{0}A\rho^{-1/2}_{0}\Pi(\rho^{1/2}B\rho^{1/2})\rho^{-1/2}_{0}\\
 & =A\rho^{-1/2}_{0}\Pi(\rho^{1/2}B\rho^{1/2})\rho^{-1/2}_{0}=A\alpha(B).
\end{align*}
Next, we show that $\Tr\rho B=\Tr\rho\alpha(B)$. Indeed,
\begin{align*}
\re\Tr\rho\alpha(B) & =\re\Tr\rho\rho^{-1/2}_{0}\Pi(\rho^{1/2}B\rho^{1/2})\rho^{-1/2}_{0}\\
 & =\re\Tr\rho_{0}\rho^{-1/2}_{0}\Pi(\rho^{1/2}B\rho^{1/2})\rho^{-1/2}_{0}\\
 & =\re\Tr\Pi(\rho^{1/2}B\rho^{1/2})=\re\Tr\rho^{1/2}B\rho^{1/2}\\
 & =\re\Tr\rho B.
\end{align*}
If $\mathcal{A}$ is a complex $*$-subalgebra, then by Theorem \ref{thm:proj},
the map $\alpha$ in (\ref{eq:CE0}) is also complex-linear.

We now turn to the general case where $\rho$ is not necessarily strictly
positive. Let $\{F_{k}\}_{k}$ be a basis of the real linear space
$\kappa\mathcal{A}s\subset\mathcal{A}$, where $\kappa=I-s$, and
define
\[
\delta=\sum_{k}F_{k}\rho F^{*}_{k}.
\]
Let $\tilde{s}=\supp\delta$, $\tilde{\kappa}=I-s-\tilde{s}$, and
\[
\sigma=\rho+\delta+\tilde{\kappa}.
\]
Then $\rho\perp\delta\perp\tilde{\kappa}$ and $\sigma>0$. We claim
that $\sigma$ satisfies condition (iii), namely, $\sigma A\sigma^{-1}\in\mathcal{A}$
for every $A\in\mathcal{A}$.

We first show that
\begin{equation}
\rho A\tilde{\kappa}=0\qquad(\forall A\in\mathcal{A}).\label{eq:rhoAk_zero}
\end{equation}
It suffices to prove that $\Tr A^{*}\rho^{2}A\tilde{\kappa}=0$. Indeed,
\begin{align}
\Tr A^{*}\rho^{2}A\tilde{\kappa} & \leq\|\rho\|\Tr A^{*}\rho A\tilde{\kappa}=\|\rho\|\Tr\kappa A^{*}\rho A\kappa\tilde{\kappa}\leq c\|\rho\|\Tr\delta\tilde{\kappa}=0,\label{eq:rhoAk_zero1}
\end{align}
where $c>0$ satisfies $\kappa A^{*}\rho A\kappa\le c\delta$. Hence
$\rho A\tilde{\kappa}=0$. By (\ref{eq:rhoAk_zero}), we have also
\[
\delta A\tilde{\kappa}=\sum_{k}F_{k}\left\{ \rho F^{*}_{k}A\tilde{\kappa}\right\} =0.
\]

We now verify that $\delta A\rho^{-1}\in\mathcal{A}$ for $A\in\A$.
Indeed,
\begin{equation}
\delta A\rho^{-1}=\sum_{k}F_{k}\left(\rho F^{*}_{k}A\rho^{-1}\right)=\sum_{k}F_{k}\left(\rho_{0}F^{*}_{k}A\rho^{-1}_{0}\right)\in\A.\label{eq:del_A_eho}
\end{equation}
The positive Hilbert--Schmidt operator $B\mapsto\delta B\rho^{-1}$
leaves $\mathcal{A}$ invariant by (\ref{eq:del_A_eho}); so does its
generalized inverse $B\mapsto\delta^{-1}B\rho$, by real polynomial
functional calculus. Taking adjoints gives
$\rho A\delta^{-1}\in\mathcal{A}$ for $A\in\mathcal{A}$.
Consequently,
\[
\delta A\delta^{-1}=\sum_{k}F_{k}\bigl(\rho F^{*}_{k}A\delta^{-1}\bigr)\in\mathcal{A}.
\]

We next show that $\tilde{s}\in\mathcal{A}$. Since there exist $a>b>0$
such that 
\[
b\sum_{k}F_{k}sF^{*}_{k}\le\sum_{k}F_{k}\rho F^{*}_{k}\le a\sum_{k}F_{k}sF^{*}_{k},
\]
we obtain $\supp\delta=\tilde{s}\in\A$. Hence $\tilde{\kappa}=I-s-\tilde{s}\in\mathcal{A}$,
and therefore $\tilde{\kappa}A\tilde{\kappa}\in\mathcal{A}.$

Combining the above facts, we conclude that $\sigma A\sigma^{-1}\in\mathcal{A}$
for $A\in\mathcal{A}$. By the already established strictly positive
case, there exists a real conditional expectation $\alpha:\mathcal{B}(\mathcal{H})\to\mathcal{A}$
such that $\re\Tr\sigma B=\re\Tr\sigma\alpha(B)$ for all $B\in\B(\H)$.
Using this, we obtain
\[
\re\Tr\rho B=\re\Tr\sigma sB=\re\Tr\sigma\alpha(sB)=\re\Tr\sigma s\alpha(B)=\re\Tr\rho\alpha(B),
\]
and the proof is complete.
\end{proof}

\begin{lem}
\label{lem:supp}Let $\rho\in\mathcal{B}(\mathcal{H})_{+}$ be a positive
operator on a Hilbert space $\mathcal{H}$, and let $\mathcal{A}\subset\mathcal{B}(\mathcal{H})$
be a real $*$-subalgebra. Let $\rho_{0}=\Pi(\rho)$, where $\Pi:\mathcal{B}(\mathcal{H})\to\mathcal{A}$
is the orthogonal projection with respect to the Hilbert--Schmidt
inner product. Then
\[
\supp\rho\le\supp\rho_{0}.
\]
If, moreover, $\supp\rho\in\mathcal{A}$, then $\supp\rho=\supp\rho_{0}$.
\end{lem}

\begin{proof}
Let $s=\supp\rho$, $s_{0}=\supp\rho_{0}$, $\kappa=I-s$, and $\kappa_{0}=I-s_{0}.$
It suffices to show that $\kappa_{0}\rho\kappa_{0}=0$. Indeed,
\[
\Pi(\kappa_{0}\rho\kappa_{0})=\kappa_{0}\Pi(\rho)\kappa_{0}=\kappa_{0}\rho_{0}\kappa_{0}=0.
\]
Since $\Pi$ is faithful, it follows that $\kappa_{0}\rho\kappa_{0}=0$,
and hence $s\le s_{0}$.

Assume now that $s\in\mathcal{A}$. Then there exists $a>0$ such
that $as-\rho\ge0$. By Theorem \ref{thm:proj} (iii), $\Pi(\rho)\le as$.
Hence $s_{0}\le s$. Together with $s\le s_{0}$, this implies $s=s_{0}$.
\end{proof}

\section{Real Jordan Algebras and Real Positive Maps\label{sec:jordan_pos}}

In this section, we discuss in more detail the properties of real
Jordan algebras $\mathcal{A}_{J}$ introduced in Definition \ref{def:jordan_alg}.

Every finite-dimensional real Jordan algebra contained in
$\mathcal{B}_{h}(\mathcal{H})$ admits the abstract decomposition
(\ref{eq:jordan_deco}) \cite{jordan1,jordan2,jordan3}.
Here $\Gamma_{n}$ denotes a spin factor. Namely, there exist $n$
Hermitian matrices $\{\gamma_{i}\}^{n}_{i=1}$ satisfying
\[
\gamma_{i}\circ\gamma_{j}=\delta_{ij}I,
\]
and $\Gamma_{n}=\Span_{\R}\bigl(\{\gamma_{i}\}^{n}_{i=1}\cup\{I\}\bigr)$.
For $\Gamma_{3}$, one may take $\gamma_{1},\gamma_{2},\gamma_{3}$
to be the usual Pauli matrices. Moreover, a representation of
$\Gamma_{n+2}$ can be constructed recursively from one of $\Gamma_{n}$.
Indeed, if $\gamma_{1},\dots,\gamma_{n}$ generate $\Gamma_{n}$, then one may define
\[
\gamma'_{i}=\begin{pmatrix}0 & \gamma_{i}\\
\gamma_{i} & 0
\end{pmatrix}\qquad(1\leq i\leq n),
\]
together with
\[
\gamma'_{n+1}=\begin{pmatrix}I & 0\\
0 & -I
\end{pmatrix},\qquad\gamma'_{n+2}=\begin{pmatrix}0 & -\ii I\\
\ii I & 0
\end{pmatrix}.
\]
Then $\gamma_{i}'\circ\gamma_{j}'=\delta_{ij}I$ for $1\le i,j\le n+2$.
Since $\Gamma_{1}\cong\mathbb{R}\oplus\mathbb{R}$, $\Gamma_{2}\cong\mathbb{M}_{2}(\mathbb{R})_{h}$,
$\Gamma_{3}\cong\mathbb{M}_{2}(\mathbb{C})_{h}$, and $\Gamma_{5}\cong\mathbb{M}_{2}(\mathbb{H})_{h}$,
only the cases $n=4$ and $n\ge6$ give additional abstract Jordan types.
For $n\ge2$, each complex irreducible representation of $\Gamma_{n}$
acts on $\mathbb{C}^{2^{\lfloor n/2\rfloor}}$ and generates the full
complex $*$-algebra $\mathbb{M}_{2^{\lfloor n/2\rfloor}}(\mathbb{C})$.
There is one such representation up to unitary equivalence for even $n$,
and two for odd $n$; a given representation may contain both odd-$n$
types. The abstract identifications above therefore do not identify
all concrete representations with multiples of the standard embeddings.

For the Jordan product, the following identity is useful. For $A,B,C\in\mathcal{B}_{h}(\mathcal{H})$,
\begin{equation}
\frac{1}{2}\left(ABC+CBA\right)=\left\{ A\circ(B\circ C)+(A\circ B)\circ C-(A\circ C)\circ B\right\} \label{eq:jordan3}
\end{equation}
In particular, if $A,B,C\in\mathcal{A}_{J}$, then $\frac{1}{2}(ABC+CBA)\in\mathcal{A}_{J}$.
Note that the definition of the square-root likelihood ratio in (\ref{eq:ratio})
and the definition of the SLD in (\ref{eq:sld}) can both be expressed
in terms of the Jordan product. By contrast, for $A,B,C,D\in\mathcal{A}_{J}$,
the inclusion $\frac{1}{2}(ABCD+DCBA)\in\mathcal{A}_{J}$ does not
hold in general.
\begin{example}
Let $\sigma_{1},\sigma_{2},\sigma_{3}\in\mathbb{M}_{2}(\mathbb{C})$
be the Pauli matrices, and define $\gamma_{1}=\sigma_{1}\otimes\sigma_{1}$,
$\gamma_{2}=\sigma_{1}\otimes\sigma_{2}$, $\gamma_{3}=\sigma_{1}\otimes\sigma_{3}$,
$\gamma_{4}=\sigma_{3}\otimes I_{2}$. Since $\gamma_{i}\circ\gamma_{j}=\delta_{ij}I$,
the space $\Gamma_{4}=\Span_{\R}\{\gamma_{1},\gamma_{2},\gamma_{3},\gamma_{4},I_{4}\}$
is a real Jordan algebra. However,
\[
\frac{1}{2}\left(\gamma_{1}\gamma_{2}\gamma_{3}\gamma_{4}+\gamma_{4}\gamma_{3}\gamma_{2}\gamma_{1}\right)=\sigma_{2}\otimes I_{2}\not\in\Gamma_{4}.
\]
\end{example}

We next prove several results concerning real Jordan algebras and
real positive maps. These results are closely related to the classical
connection between positive projections and Jordan structure developed
by Effros and Størmer \cite{effros}.
\begin{lem}
\label{lem:H2H}Let $\alpha:\mathcal{B}(\mathcal{H})\to\mathcal{B}(\mathcal{H})$
be a real positive map. If $B\in\mathcal{B}_{h}(\mathcal{H})$, then
$\alpha(B)\in\mathcal{B}_{h}(\mathcal{H})$.
\end{lem}

\begin{proof}
Write $B=P-N$ with $P\ge0$ and $N\ge0$. Then $\alpha(B)=\alpha(P)-\alpha(N),$
and since $\alpha(P)\ge0$ and $\alpha(N)\ge0$, it follows that $\alpha(B)$
is self-adjoint.
\end{proof}

\begin{lem}[Schwarz inequality]
\label{lem:schwartz}Let $\alpha:\mathcal{B}(\mathcal{H})\to\mathcal{B}(\mathcal{H})$
be a real unital positive map. Then for every $B\in\mathcal{B}_{h}(\mathcal{H})$,
\[
\alpha(B^{2})\ge\alpha(B)^{2}.
\]
\end{lem}

\begin{proof}
Write $B=\sum_j b_j P_j$ in spectral form and put $C=\alpha(B)$.
Since $\alpha(P_j)\ge0$ and $\sum_j\alpha(P_j)=I$,
\[
\alpha(B^2)-\alpha(B)^2
=\sum_j(b_jI-C)\alpha(P_j)(b_jI-C)\ge0.
\]
\end{proof}

\begin{thm}
\label{thm:pos2J}Let $\alpha:\mathcal{B}(\mathcal{H})\to\mathcal{B}(\mathcal{H})$
be a real unital positive map. Assume that $\alpha$ is faithful and
idempotent, i.e. $\alpha\circ\alpha=\alpha$. Then $\mathcal{A}_{J}:=\alpha(\mathcal{B}_{h}(\mathcal{H}))$
is a real Jordan algebra. Moreover, the restriction of $\alpha$ to
$\mathcal{B}_{h}(\mathcal{H})$ is a real Jordan conditional expectation.
That is, for every $A\in\mathcal{A}_{J}$ and $B\in\mathcal{B}_{h}(\mathcal{H})$,
\begin{align}
\alpha(A\circ B) & =A\circ\alpha(B),\label{eq:JordanCE}
\end{align}
and
\begin{equation}
A\alpha(B)A=\alpha(ABA).\label{eq:JordanCE2}
\end{equation}
\end{thm}

\begin{proof}
We first show that $\alpha(A^{2})=A^{2}$ for every $A\in\mathcal{A}_{J}$.
By Lemma \ref{lem:schwartz},
\[
\alpha(A^{2})-\alpha(A)^{2}=\alpha(A^{2})-A^{2}\ge0.
\]
On the other hand,
\[
\alpha\bigl(\alpha(A^{2})-A^{2}\bigr)=0.
\]
Since $\alpha$ is faithful, it follows that $\alpha(A^{2})=A^{2}$.

Now let $A,B\in\mathcal{A}_{J}$. We show that $A\circ B\in\mathcal{A}_{J}$.
Since
\[
2(A\circ B)=(A+B)^{2}-A^{2}-B^{2},
\]
we have
\[
2\alpha(A\circ B)=\alpha((A+B)^{2})-\alpha(A^{2})-\alpha(B^{2}).
\]
Using $\alpha((A+B)^{2})=(A+B)^{2}$, $\alpha(A^{2})=A^{2}$, and
$\alpha(B^{2})=B^{2}$, we obtain $2\alpha(A\circ B)=2(A\circ B)$.
Hence $A\circ B\in\mathcal{A}_{J}$, and therefore $\mathcal{A}_{J}$
is a real Jordan algebra.

Next we prove (\ref{eq:JordanCE}). By Lemma \ref{lem:schwartz},
for every $t\in\mathbb{R}$,
\[
\alpha\bigl((A+tB)^{2}\bigr)\ge\alpha(A+tB)^{2}.
\]
Expanding both sides gives
\[
\alpha(A^{2})+t^{2}\alpha(B^{2})+2t\,\alpha(A\circ B)\ge A^{2}+t^{2}\alpha(B)^{2}+2t\,A\circ\alpha(B).
\]
Since $\alpha(A^{2})=A^{2}$, for $t>0$ we obtain
\[
t\,\alpha(B^{2})+2\alpha(A\circ B)\ge t\,\alpha(B)^{2}+2A\circ\alpha(B).
\]
Letting $t\downarrow0$, we get $\alpha(A\circ B)\ge A\circ\alpha(B)$.
Similarly, applying the same argument for $t<0$ and letting $t\uparrow0$,
we obtain $\alpha(A\circ B)\le A\circ\alpha(B)$. Therefore, $\alpha(A\circ B)=A\circ\alpha(B)$. 

Finally, (\ref{eq:JordanCE2}) follows from (\ref{eq:jordan3}). Indeed,
\[
ABA=2\,A\circ(A\circ B)-A^{2}\circ B,
\]
so applying (\ref{eq:JordanCE}) twice and using $\alpha(A^{2})=A^{2}$,
we get
\[
\alpha(ABA)=2\,A\circ(A\circ\alpha(B))-A^{2}\circ\alpha(B)=A\,\alpha(B)\,A.
\]
\end{proof}

At the end of this section, we record that the orthogonal projection
onto a real Jordan algebra is a faithful positive map.
\begin{lem}
\label{lem:proj-J}Let $\mathcal{A}_{J}\subset\mathcal{B}_{h}(\mathcal{H})$
be a real Jordan algebra, and let $\Pi_{J}:\mathcal{B}_{h}(\mathcal{H})\to\mathcal{A}_{J}$
be the orthogonal projection with respect to the real Hilbert--Schmidt
inner product. Then the following hold:
\begin{description}
\item [{(i)}] $\ensuremath{\Pi_{J}}$ is positive. That is, if $B\in\mathcal{B}(\mathcal{H})_{+}$,
then $\Pi_{J}(B)\ge0$.
\item [{(ii)}] $\ensuremath{\Pi_{J}}$ is faithful. That is, if $B\ge0$
and $\Pi_{J}(B)=0$, then $B=0$.
\end{description}
\end{lem}

\begin{proof}
We first prove (i). Let $B\in\mathcal{B}(\mathcal{H})_{+}$. Since
$\Pi_{J}(B)$ is self-adjoint, write its spectral decomposition as
$\Pi_{J}(B)=\sum_{i}\lambda_{i}P_{i},$ where $\lambda_{i}\in\mathbb{R}$
and $P_{i}\in\mathcal{A}_{J}$ are the spectral projections. Then
\[
\Tr P_{i}\Pi_{J}(B)=\langle P_{i},\Pi_{J}(B)\rangle=\langle P_{i},B\rangle=\Tr P_{i}B\ge0.
\]
Since $\Tr P_{i}\Pi_{J}(B)=\lambda_{i}\Tr P_{i},$ it follows that
$\lambda_{i}\ge0$ for all $i$, and hence $\Pi_{J}(B)\ge0$.

We next prove (ii). Let $B\ge0$ and assume that $\Pi_{J}(B)=0$.
Then
\[
\Tr B=\langle I,B\rangle=\langle I,\Pi_{J}(B)\rangle=\Tr\Pi_{J}(B)=0.
\]
Since $B\ge0$, this implies $B=0$. Therefore $\Pi_{J}$ is faithful.
\end{proof}

\section{Sufficient algebras derived from a sufficient real positive map\label{sec:pos2ce}}

The goal of this section is to prove Theorem \ref{thm:pos_to_CE}. 

Let $\mathcal{S}$ be a model, and let $\alpha:\mathcal{B}(\mathcal{H}_{\mathcal{S}})\to\mathcal{B}(\mathcal{H}_{\mathcal{S}})$
be a sufficient real unital positive map. Note that, unlike in the
conventional setting, a model in our sense is not restricted to families
of states or positive operators. By Lemma \ref{lem:H2H}, we may consider
the real-linear map $\alpha_{h}:\mathcal{B}_{h}(\mathcal{H}_{\mathcal{S}})\to\mathcal{B}_{h}(\mathcal{H}_{\mathcal{S}}),\,B\mapsto\alpha(B).$

\begin{lem}
\label{lem:faithful} The map $\alpha_{h}$ is faithful. That is,
if $B\in\mathcal{B}(\mathcal{H}_{\mathcal{S}})_{+}$ and $\alpha_{h}(B)=0$,
then $B=0$.
\end{lem}

\begin{proof}
Let $\kappa_{\alpha}=\sup\left\{ B\in\mathcal{B}(\mathcal{H}_{\mathcal{S}})_{+}\mid\|B\|\le1,\ \alpha_{h}(B)=0\right\} $,
and $s_{\alpha}=I-\kappa_{\alpha}.$ Assume $\kappa_{\alpha}\neq0$,
and derive a contradiction. Decompose $\H_{\S}$ so that $s_{\alpha}=\begin{pmatrix}I & 0\\
0 & 0
\end{pmatrix}.$ 

We first show that $\alpha_{h}\begin{pmatrix}I & 0\\
0 & 0
\end{pmatrix}=I.$ Indeed, this follows immediately from $\alpha_{h}(I)=I$ and $\alpha_{h}\begin{pmatrix}0 & 0\\
0 & I
\end{pmatrix}=0.$

Next we show that, for every $B=\begin{pmatrix}0 & B^{*}_{1}\\
B_{1} & B_{2}
\end{pmatrix}\in\mathcal{B}_{h}(\mathcal{H}_{\mathcal{S}}),$ one has $\alpha_{h}(B)=0$. For any $\varepsilon>0$, there exists
$\lambda>0$ such that
\[
\begin{pmatrix}\varepsilon I & 0\\
0 & \lambda I
\end{pmatrix}\pm\begin{pmatrix}0 & B^{*}_{1}\\
B_{1} & B_{2}
\end{pmatrix}>0.
\]
Hence
\[
\alpha_{h}\!\left(\begin{pmatrix}\varepsilon I & 0\\
0 & \lambda I
\end{pmatrix}\pm\begin{pmatrix}0 & B^{*}_{1}\\
B_{1} & B_{2}
\end{pmatrix}\right)=\varepsilon I\pm\alpha_{h}\!\left(\begin{pmatrix}0 & B^{*}_{1}\\
B_{1} & B_{2}
\end{pmatrix}\right)\ge0.
\]
Since $\varepsilon>0$ is arbitrary, it follows that $\alpha_{h}\!\left(\begin{pmatrix}0 & B^{*}_{1}\\
B_{1} & B_{2}
\end{pmatrix}\right)=0.$

Now let $X\in\mathcal{S}$, and write $X=\begin{pmatrix}X_{0} & X^{*}_{1}\\
X_{1} & X_{2}
\end{pmatrix}$. We show that $X_{1}=0$ and $X_{2}=0$. Indeed, for every $B=\begin{pmatrix}0 & B^{*}_{1}\\
B_{1} & B_{2}
\end{pmatrix}\in\mathcal{B}_{h}(\mathcal{H}_{\mathcal{S}})$, we have
\[
\Tr XB=\Tr X\alpha_{h}(B)=0.
\]
Hence $X_{1}=0$ and $X_{2}=0$. This contradicts the definition $\mathcal{H}_{\mathcal{S}}=\Span_{\mathbb C}\{X(\mathcal{H})\mid X\in\mathcal{S}\}$.
Therefore $\kappa_{\alpha}=0$, and $\alpha_{h}$ is faithful.
\end{proof}

\begin{lem}
\label{lem:pos_map_norm} $\left\Vert \alpha_{h}\right\Vert =1$. 
\end{lem}

\begin{proof}
Let $B\in\mathcal{B}_{h}(\mathcal{H}_{\mathcal{S}})$ satisfy $\|B\|\le1$.
By the Schwarz inequality in Lemma \ref{lem:schwartz},
\[
\alpha_{h}(B)^{2}\le\alpha_{h}(B^{2})\le\alpha_{h}(I)=I.
\]
Hence $\|\alpha_{h}(B)\|\le1$, and therefore $\|\alpha_{h}\|=1$.
\end{proof}

\begin{lem}
\label{lem:suff_proj} Let $\mathcal{A}_{J}=\{A\in\mathcal{B}_{h}(\mathcal{H}_{\mathcal{S}})\mid\alpha_{h}(A)=A\}$
be the fixed-point set of $\alpha_{h}$. Then
\[
\beta_{J}=\lim_{n\to\infty}\left\{ \frac{1}{2}(\mathrm{id}+\alpha_{h})\right\} ^{n}
\]
exists, and $\beta_{J}:\mathcal{B}_{h}(\mathcal{H}_{\mathcal{S}})\to\mathcal{B}_{h}(\mathcal{H}_{\mathcal{S}})$
is a sufficient positive map satisfying $\beta^{2}_{J}=\beta_{J}$
and $\beta_{J}(\mathcal{B}_{h}(\mathcal{H}_{\mathcal{S}}))=\mathcal{A}_{J}$.
In other words, $\beta_{J}$ is a real projection onto $\mathcal{A}_{J}$.
\end{lem}

\begin{proof}
By Lemma \ref{lem:pos_map_norm}, every eigenvalue $\lambda$ of $\alpha_{h}$
satisfies $|\lambda|\le1$. Hence every eigenvalue of $\frac{1}{2}(\mathrm{id}+\alpha_{h})$
with modulus one must be equal to $1$. Again by Lemma \ref{lem:pos_map_norm},
$\sup_{n\in\mathbb{N}}\left\Vert \left\{ \frac{1}{2}(\mathrm{id}+\alpha_{h})\right\} ^{n}\right\Vert \le1.$
Since $\alpha_{h}$ and $\frac{1}{2}(\mathrm{id}+\alpha_{h})$ have
the same fixed points, the conclusion follows from Lemma \ref{lem:lim_proj}.
\end{proof}

The same conclusion also holds if one uses the ergodic average $\beta_{E}=\lim_{n\to\infty}\frac{1}{n}\sum^{n}_{k=1}\alpha^{k}_{h}$. 
\begin{lem}
\label{lem:lim_proj}Let $V$ be a finite-dimensional real vector
space, and let $T:V\to V$ be a linear map whose eigenvalues satisfy
$|\lambda|\le1$, with $1$ being the only eigenvalue of modulus one.
If $\sup_{n\in\mathbb{N}}\|T^{n}\|<\infty$, then the limit $T^{\infty}=\lim_{n\to\infty}T^{n}$
exists, and $T^{\infty}$ is the projection onto the fixed-point subspace
of $T$.
\end{lem}

Define a positive operator $\omega\in\mathcal{B}(\mathcal{H}_{\mathcal{S}})_{+}$
by 
\[
\Tr\omega B=\Tr\beta_{J}(B),\qquad B\in\mathcal{B}_{h}(\mathcal{H}_{\mathcal{S}}).
\]

\begin{lem}
\label{lem:omega}The operator $\omega$ has the following properties:
\begin{description}
\item [{(i)}] $\omega>0$.
\item [{(ii)}] For every $B\in\mathcal{B}_{h}(\mathcal{H}_{\mathcal{S}})$,
\[
\re\Tr\omega B=\re\Tr\omega\beta_{J}(B).
\]
\item [{(iii)}] For every $A\in\mathcal{A}_{J}$, 
\[
\omega A=A\omega.
\]
\end{description}
\end{lem}

\begin{proof}
We first prove (i). Since $\beta_{J}$ is faithful by Lemma \ref{lem:faithful},
it follows immediately that $\omega>0$. 

For (ii), we compute
\[
\re\Tr\omega\beta_{J}(B)=\re\Tr\beta_{J}(\beta_{J}(B))=\re\Tr\beta_{J}(B)=\re\Tr\omega B.
\]

We next prove (iii). Let $\mathcal{A}^{e}_{J}=\{A\in\mathcal{A}_{J}\mid A^{2}=I\}$.
Since $\mathcal{A}_{J}=\Span_{\mathbb{R}}\mathcal{A}^{e}_{J}$, it
suffices to show that $\omega A=A\omega$, equivalently $A\omega A=\omega$,
for every $A\in\mathcal{A}^{e}_{J}$. For this, it suffices to prove
that $\Tr A\omega AC=\Tr\omega C$ for all $C\in\mathcal{B}_{h}(\mathcal{H}_{\mathcal{S}})$.
Using (\ref{eq:JordanCE2}), we obtain
\begin{align*}
\Tr A\omega AC & =\Tr\omega ACA=\Tr\beta_{J}(ACA)=\Tr A\beta_{J}(C)A\\
 & =\Tr A^{2}\beta_{J}(C)=\Tr\beta_{J}(C)=\Tr\omega C.
\end{align*}
This proves the claim.
\end{proof}

Let $\mathcal{A}_{\mathbb{R}}$ and $\mathcal{A}_{\mathbb{C}}$ be
the real and complex $*$-subalgebras generated by $\mathcal{A}_{J}$,
respectively. Let $\Pi_{J}:\mathcal{B}(\mathcal{H}_{\mathcal{S}})\to\mathcal{A}_{J}$,
$\Pi_{\mathbb{R}}:\mathcal{B}(\mathcal{H}_{\mathcal{S}})\to\mathcal{A}_{\mathbb{R}}$,
$\Pi_{\mathbb{C}}:\mathcal{B}(\mathcal{H}_{\mathcal{S}})\to\mathcal{A}_{\mathbb{C}}$
be the orthogonal projections with respect to the real Hilbert--Schmidt
inner product.
\begin{lem}
\label{lem:omega_proj}The orthogonal projections $\Pi_{J}$, $\Pi_{\mathbb{R}}$,
and $\Pi_{\mathbb{C}}$ satisfy the following properties:
\begin{description}
\item [{(i)}] $\Pi_{J}(\omega)=I$. For $\mathbb{K}=\mathbb{R},\mathbb{C}$,
$z_{\mathbb{K}}:=\Pi_{\mathbb{K}}(\omega)$ is strictly positive
and belongs to the center of $\mathcal{A}_{\mathbb{K}}$.
\item [{(ii)}] For every $X\in\mathcal{S}$,
\[
\Pi_{J}(X)=X\omega^{-1},\qquad
\Pi_{\mathbb{K}}(X)=z_{\mathbb{K}}\Pi_{J}(X)
\quad(\mathbb{K}=\mathbb{R},\mathbb{C}).
\]
\item [{(iii)}] For every $X\in\mathcal{S}$,
\[
[\omega,X]=0.
\]
\end{description}
\end{lem}

\begin{proof}
We first prove (i). For $A\in\mathcal{A}_{J}$,
$\Tr\omega A=\Tr\beta_{J}(A)=\Tr A$, so $\Pi_{J}(\omega)=I$.
Since $\omega$ commutes with $\mathcal{A}_{\mathbb{K}}$ and
$\Pi_{\mathbb{K}}$ is a bimodule map,
\[
A z_{\mathbb{K}}=\Pi_{\mathbb{K}}(A\omega)
=\Pi_{\mathbb{K}}(\omega A)=z_{\mathbb{K}}A
\qquad(A\in\mathcal{A}_{\mathbb{K}}).
\]
Positivity and unitality of $\Pi_{\mathbb{K}}$, together with $\omega>0$,
give $z_{\mathbb{K}}>0$.

We next prove (ii). For $X\in\mathcal{S}$ and
$B\in\mathcal{B}_{h}(\mathcal{H}_{\mathcal{S}})$,
\begin{align*}
\re\Tr\omega\,\Pi_{J}(X)B
&=\re\Tr\beta_{J}(\Pi_{J}(X)\circ B)\\
&=\re\Tr\Pi_{J}(X)\circ\beta_{J}(B)
=\re\Tr X\beta_{J}(B)=\re\Tr XB.
\end{align*}
Since $\omega$ commutes with $\Pi_{J}(X)$, this proves
$X=\omega\Pi_{J}(X)$. Applying $\Pi_{\mathbb{K}}$ and using its
bimodule property gives
$\Pi_{\mathbb{K}}(X)=z_{\mathbb{K}}\Pi_{J}(X)$.
Finally, (iii) follows from $X=\omega\Pi_{J}(X)$ and
$[\omega,\Pi_{J}(X)]=0$.
\end{proof}

Define maps $\beta_{\mathbb{R}}:\mathcal{B}(\mathcal{H}_{\mathcal{S}})\to\mathcal{A}_{\mathbb{R}}$
and $\beta_{\mathbb{C}}:\mathcal{B}(\mathcal{H}_{\mathcal{S}})\to\mathcal{A}_{\mathbb{C}}$
by
\[
\beta_{\mathbb{K}}(B)=z_{\mathbb{K}}^{-1}
\Pi_{\mathbb{K}}(\omega^{1/2}B\omega^{1/2}),
\qquad\mathbb{K}=\mathbb{R},\mathbb{C}.
\]

\begin{lem}
\label{lem:pos2ce}The following hold: The maps $\beta_{\mathbb{R}}$
and $\beta_{\mathbb{C}}$ satisfy the following properties:
\begin{description}
\item [{(i)}] $\beta_{\mathbb{R}}$ is a real conditional expectation,
and $\beta_{\mathbb{C}}$ is a complex conditional expectation.
\item [{(ii)}] For every $B\in\mathcal{B}(\mathcal{H}_{\mathcal{S}})$
and every $X\in\mathcal{S}\cup\{\omega\}$,
\[
\re\Tr X\beta_{\mathbb{R}}(B)=\re\Tr X\beta_{\mathbb{C}}(B)=\re\Tr XB.
\]
\end{description}
\end{lem}

\begin{proof}
Fix $\mathbb{K}=\mathbb{R}$ or $\mathbb{C}$. Since
$z_{\mathbb{K}}>0$ is central, $\beta_{\mathbb{K}}$ is positive.
For $A\in\mathcal{A}_{\mathbb{K}}$,
\[
\beta_{\mathbb{K}}(A)=z_{\mathbb{K}}^{-1}
\Pi_{\mathbb{K}}(\omega A)=A.
\]
The commutation of $\omega$ with $\mathcal{A}_{\mathbb{K}}$ and
the bimodule property of $\Pi_{\mathbb{K}}$ give
$\beta_{\mathbb{K}}(AB)=A\beta_{\mathbb{K}}(B)$.
It also preserves the involution, since $z_{\mathbb{K}}$ is central.
Thus $\beta_{\mathbb{K}}$ is a conditional expectation.
It is faithful because $\omega>0$ and $\Pi_{\mathbb{K}}$ is
positive and preserves the trace of positive operators.

For (ii), let $X\in\mathcal{S}\cup\{\omega\}$ and write
$Y=X\omega^{-1}\in\mathcal{A}_{J}$. Lemma \ref{lem:omega_proj}
gives $\Pi_{\mathbb{K}}(X)=z_{\mathbb{K}}Y$ (also for $X=\omega$).
Hence, for every $B\in\mathcal{B}(\mathcal{H}_{\mathcal{S}})$,
\begin{align*}
\re\Tr X\beta_{\mathbb{K}}(B)
&=\re\Tr z_{\mathbb{K}}Y\,z_{\mathbb{K}}^{-1}
\Pi_{\mathbb{K}}(\omega^{1/2}B\omega^{1/2})\\
&=\re\Tr Y\omega^{1/2}B\omega^{1/2}
=\re\Tr XB.
\end{align*}
\end{proof}

This completes the proof of Theorem \ref{thm:pos_to_CE}.

\section{Minimal Sufficient Real $*$-Subalgebras and Real Jordan Algebras
via the Likelihood-Ratio Set\label{sec:min_suff}}

The goal of this section is to prove the main Theorems \ref{thm:min_suff_alg}
and \ref{thm:min_suff_J}, which characterize minimal sufficient real
$*$-subalgebras and real Jordan algebras in terms of the likelihood-ratio
set.

For an absolutely continuous model $\mathcal{S}$, let $\rho\in\mathcal{S}$
be a reference state. Then the square-root likelihood ratio can be
written as
\[
R_{X}=X^{1/2}(X^{1/2}\rho X^{1/2})^{-1/2}X^{1/2}+\sigma
\]
for $X\in\mathcal{S}$ \cite{qLAN2}. Here, for $Z\ge0$, the operator
$Z^{-1/2}$ denotes the generalized inverse of $Z^{1/2}$, and $\sigma\ge0$
is an arbitrary positive operator satisfying $\rho\sigma=0$.
We choose $\sigma=0$, which gives the unique solution of minimal
Hilbert--Schmidt norm. The following holds for sufficient real Jordan algebras. 
\begin{thm}
\label{thm:ratio_abs}Let $\mathcal{S}$ be an absolutely continuous
model, and let $\mathcal{A}_{J}$ admit a sufficient real Jordan
conditional expectation $\beta:\mathcal{B}_{h}(\mathcal{H}_{\mathcal{S}})\to\mathcal{A}_{J}$
and supporting operator $\omega>0$. Let $\rho\in\mathcal{S}$ be
the reference state. Then
\[
R_{X}\in\mathcal{A}_{J}\qquad(\forall X\in\mathcal{S}).
\]
\end{thm}

\begin{proof}
Put $X_{0}=X\omega^{-1}$ and $\rho_{0}=\rho\omega^{-1}$. Then
\[
R_{X}=X^{1/2}_{0}(X^{1/2}_{0}\rho_{0}X^{1/2}_{0})^{-1/2}X^{1/2}_{0}.
\]
By Lemma \ref{lem:omega_proj}, we have $X_{0},\rho_{0}\in\mathcal{A}_{J}$.
Hence $R_{X}\in\mathcal{A}_{J}$.
\end{proof}

For a local model $\mathcal{S}$, let $\rho\in\mathcal{S}$ be a reference
state and let $X\in\mathcal{S}$. Then the corresponding SLD can be
written as
\[
L_{X}=J^{-1}_{\rho}(X)+\sigma.
\]
Here $J_{\rho}:\mathcal{B}(\mathcal{H}_{\mathcal{S}})\to\mathcal{B}(\mathcal{H}_{\mathcal{S}})$
is defined by
\[
J_{\rho}(Z)=\rho\circ Z,
\]
and is a positive operator with respect to the Hilbert--Schmidt inner
product. The map $J^{-1}_{\rho}$ denotes the generalized inverse
of $J_{\rho}$ with respect to the Hilbert--Schmidt inner product,
and $\sigma$ is an arbitrary Hermitian operator satisfying $\rho\sigma=0$.
Again, we choose $\sigma=0$, which gives the unique solution of
minimal Hilbert--Schmidt norm.
The following holds for sufficient real Jordan algebras.
\begin{thm}
\label{thm:ratio_local}Let $\mathcal{S}$ be a local model, and let
$\mathcal{A}_{J}$ admit a sufficient real Jordan conditional
expectation $\beta:\mathcal{B}_{h}(\mathcal{H}_{\mathcal{S}})\to\mathcal{A}_{J}$
and supporting operator $\omega>0$. Let $\rho\in\mathcal{S}$ be
the reference state. Then
\[
L_{X}\in\mathcal{A}_{J}\qquad(\forall X\in\mathcal{S}).
\]
\end{thm}

\begin{proof}
Put $X_{0}=X\omega^{-1}$ and $\rho_{0}=\rho\omega^{-1}$. Then
\[
L_{X}=J^{-1}_{\rho_{0}}(X_{0}).
\]
By Lemma \ref{lem:omega_proj}, we have $X_{0},\rho_{0}\in\mathcal{A}_{J}$.
Hence $L_{X}\in\mathcal{A}_{J}$.
\end{proof}

For an absolutely continuous model, let $\mathcal{R}=\{R_{X}\mid X\in\mathcal{S}\}$,
and for a local model, let $\mathcal{R}=\{L_{X}\mid X\in\mathcal{S}\}$.
We now prove the main Theorem \ref{thm:min_suff_alg} concerning minimal
sufficient real $*$-subalgebras and the likelihood-ratio set $\mathcal{R}$.
\begin{proof}[Proof of Theorem \ref{thm:min_suff_alg} ]
Let $\mathcal{A}\subset\mathcal{B}(\mathcal{H}_{\mathcal{S}})$ be
a real $*$-subalgebra satisfying the conditions (i) $\rho A\rho^{-1}\in\mathcal{A}$
for every $A\in\mathcal{A}$, and (ii) $\mathcal{R}\subset\mathcal{A}$.
We show that $\mathcal{A}$ admits a sufficient real conditional expectation.
By Theorem \ref{thm:modular}, condition (i) implies that there exists
a real conditional expectation $\alpha:\mathcal{B}(\mathcal{H}_{\mathcal{S}})\to\mathcal{A}$
such that
\[
\re\Tr\rho B=\re\Tr\rho\,\alpha(B)\qquad(\forall B\in\mathcal{B}(\mathcal{H}_{\mathcal{S}})).
\]
If $\mathcal{A}$ is a complex $*$-subalgebra, then $\alpha$ may
be chosen to be a complex conditional expectation.

If $\mathcal{S}$ is an absolutely continuous model, then by condition
(ii), for every $X\in\mathcal{S}$ and $B\in\mathcal{B}(\mathcal{H}_{\mathcal{S}})$,
\begin{align*}
\re\Tr XB & =\re\Tr\rho R_{X}BR_{X}=\re\Tr\rho\,\alpha(R_{X}BR_{X})\\
 & =\re\Tr\rho\,R_{X}\alpha(B)R_{X}=\re\Tr X\alpha(B).
\end{align*}
If $\mathcal{S}$ is a local model, then by condition (ii), for every
$X\in\mathcal{S}$ and $B\in\mathcal{B}(\mathcal{H}_{\mathcal{S}})$,
\begin{align*}
\re\Tr XB & =\re\Tr(\rho\circ L_{X})B=\re\Tr\rho(B\circ L_{X})\\
 & =\re\Tr\rho\,\alpha(B\circ L_{X})=\re\Tr\rho(\alpha(B)\circ L_{X})\\
 & =\re\Tr(\rho\circ L_{X})\alpha(B)\\
 & =\re\Tr X\alpha(B).
\end{align*}
Thus $\mathcal{A}$ is a sufficient real $*$-subalgebra, and $\alpha$
is a sufficient real conditional expectation onto $\mathcal{A}$.

Finally, let $\mathcal{D}$ be any sufficient real $*$-subalgebra
and choose a sufficient map with range in $\mathcal{D}$.
By Theorem \ref{thm:pos_to_CE}, its Hermitian fixed-point algebra
$K\subset\mathcal{D}_h$ generates a sufficient real $*$-algebra
$\mathcal{B}\subset\mathcal{D}$, and $\rho=\omega q$ for
$q\in K$ and $\omega>0$ commuting with $\mathcal{B}$.
Let $q^{-1}$ denote its generalized inverse.
Thus $\rho A\rho^{-1}=qAq^{-1}\in\mathcal{B}$ for $A\in\mathcal{B}$,
and Theorems \ref{thm:ratio_abs} and \ref{thm:ratio_local} give
$\mathcal{R}\subset K$. Hence $\mathcal{B}$ satisfies (i) and (ii),
so the smallest such algebra is contained in every sufficient
$\mathcal{D}$. This proves minimality. For a sufficient complex
$\mathcal{D}$, use the complex $*$-algebra generated by $K$ instead.
\end{proof}

In the work \cite{qLAN3}, which established the asymptotic representation
theorem in q-LAN, the $\D$-invariant extension of square-root likelihood
ratios played a crucial role. For a complex $*$-subalgebra containing $\supp\rho$,
condition (i) in Theorem \ref{thm:min_suff_alg} is equivalent
to $\D$-invariance. See Appendix \ref{sec:dinv} for details. The
following corollary may also be useful.

\begin{cor}
Let
\[
\mathcal{A}_{0}=\left\{ \rho^{n}R\rho^{-n}\mid R\in\mathcal{R},\ n\in\mathbb{N}\cup\{0\}\right\} \cup\left\{ \rho^{n}R_{1}\kappa R_{2}\rho^{-n}\mid R_{1},R_{2}\in\mathcal{R},\ n\in\mathbb{N}\right\} ,
\]
where $\kappa=I-\supp\rho$. Then the real $*$-subalgebra generated
by $\mathcal{A}_{0}$ is the minimal sufficient real $*$-subalgebra.
In particular, if $\rho>0$ or $\rho$ is a pure state, then the real
$*$-subalgebra generated by
\[
\mathcal{A}_{0}=\left\{ \rho^{n}R\rho^{-n}\mid R\in\mathcal{R},\ n\in\mathbb{N}\cup\{0\}\right\} 
\]
is minimal sufficient. 
\end{cor}

\begin{proof}
It is straightforward to verify that the real $*$-subalgebra $\mathcal{A}$
generated by $\mathcal{A}_{0}$ is the smallest real $*$-subalgebra
satisfying conditions (i) and (ii) of Theorem \ref{thm:min_suff_alg},
namely, 
\[
\mathcal{R}\subset\mathcal{A},\qquad\rho\mathcal{A}\rho^{-1}\subset\mathcal{A}.
\]
\end{proof}

We now prove Theorem \ref{thm:min_suff_J}, which characterizes the minimal sufficient
real Jordan algebra using the likelihood-ratio set $\mathcal{R}$.
\begin{proof}[Proof of Theorem \ref{thm:min_suff_J}]
Let $E$ be a sufficient real conditional expectation onto
$\mathcal{A}$, and let $\alpha$ be a sufficient map with Hermitian
range in an arbitrary sufficient real Jordan algebra $L$.
By Theorem \ref{thm:pos_to_CE}, the Hermitian fixed-point set $K$
of $\gamma=(\alpha+E)/2$ is a sufficient Jordan algebra admitting
a conditional expectation. Since $B^2\in K$ whenever $B\in K$,
Lemma \ref{lem:schwartz} gives
\[
0=\gamma_h(B^2)-B^2
\ge\frac14\bigl(\alpha_h(B)-E_h(B)\bigr)^2,\qquad B\in K.
\]
Thus $\alpha_h(B)=E_h(B)=B$, so $K\subset L\cap\mathcal{A}_h$.
The real $*$-algebra generated by $K$ is sufficient by
Theorem \ref{thm:pos_to_CE}, and therefore equals $\mathcal{A}$
by minimality. Also, Theorems \ref{thm:ratio_abs} and
\ref{thm:ratio_local} give $\mathcal{R}\subset K$.

Theorem \ref{thm:pos_to_CE} gives $\rho=\omega q$ with $q\in K$
and $\omega>0$ commuting with $\mathcal{A}$.
Write $K\cong\bigoplus_{\alpha}K_{\alpha}$ as a sum of simple
Jordan factors, with units $e_{\alpha}$, which are central in
$\mathcal{A}$. Each $\mathcal{H}_{\lambda}$ carries an irreducible
representation $\pi_{\lambda}$ of one $K_{\alpha}$, since $K$
generates $\mathcal{A}$. Schur's lemma gives
\[
P_{\lambda}\omega P_{\lambda}=w_{\lambda}P_{\lambda},\qquad
\left.P_{\lambda}\rho P_{\lambda}\right|_{\mathcal{H}_{\lambda}}
=w_{\lambda}\pi_{\lambda}(q_{\alpha}),\qquad w_{\lambda}>0.
\]
If $\Tr(P_{\lambda}\rho P_{\lambda})=0$, positivity implies
$\rho P_{\lambda}=0$. Since $P_{\lambda}$ commutes with
$\mathcal{R}$, every $X=R_X\rho R_X$ or $X=\rho\circ L_X$
would annihilate $\mathcal{H}_{\lambda}$, contradicting the
definition of $\mathcal{H}_{\mathcal{S}}$.
The irreducible representations of a fixed simple Jordan factor
have the same ordinary matrix trace functional \cite{jordan3}.
Hence $u_{\alpha}:=\Tr\pi_{\lambda}(q_{\alpha})>0$ is independent
of $\lambda$ belonging to $\alpha$. With
$z=\sum_{\alpha}u_{\alpha}e_{\alpha}$, blockwise normalization gives
\[
\rho_1=z^{-1}q\in K,\qquad
\rho=W\rho_1,\quad W:=\omega z>0,\quad [W,\mathcal{A}]=0.
\]
Consequently, $J:=\mathcal{A}_{J}\subset K\subset L$.

It remains to prove sufficiency of $J$. Every model operator is
$X=WY_X$, where $Y_X=R_X\rho_1R_X$ or $Y_X=\rho_1\circ L_X$
belongs to $J$. Put $P=\Pi_J$ and $c=P(W)>0$.
By Lemma \ref{lem:proj-J} and Theorem \ref{thm:pos2J}, $P$ is
a Jordan conditional expectation. For $s\in J$ with $s^2=I$,
$scs=P(sWs)=c$ by (\ref{eq:JordanCE2}); such elements span $J$,
so $c$ is central in $J$. Thus
\[
F(B)=c^{-1}P(W^{1/2}BW^{1/2}),\qquad
B\in\mathcal{B}_h(\mathcal{H}_{\mathcal{S}}),
\]
is positive, unital, and takes values in $J$. Since $P(X)=cY_X$,
\[
\Tr XF(B)=\Tr Y_XP(W^{1/2}BW^{1/2})=\Tr XB.
\]
The map $F$ fixes $J$ and is faithful, so it is a Jordan conditional
expectation by Theorem \ref{thm:pos2J}. Its extension
$B\mapsto F((B+B^*)/2)$ proves sufficiency.
Together with $J\subset L$ for arbitrary sufficient $L$, this
proves minimality.
\end{proof}

\section{Real Jordan Koashi--Imoto decomposition\label{sec:KI}}

In this section, we prove Theorem \ref{thm:KI_deco}. For sufficient
complex algebras, the Koashi--Imoto decomposition is well known \cite{KI-deco1,KI-deco2}.
Here we further extend it to sufficient real $*$-subalgebras and
sufficient real Jordan algebras. As before, the model $\mathcal{S}$
is not assumed to be a family of states; it may be any family of self-adjoint
operators.

We first discuss the Koashi--Imoto decomposition for sufficient real
$*$-subalgebras. Let $\mathcal{A}\subset\mathcal{B}(\mathcal{H}_{\mathcal{S}})$
admit a sufficient real conditional expectation for the model
$\mathcal{S}$. Write its represented decomposition as in
(\ref{eq:real_decomp}). The following statement applies, in particular,
to minimal sufficient real or complex $*$-subalgebras that admit
such an expectation.
\begin{thm}
\label{thm:KI_real_alg}After a unitary change of basis, every
$X\in\mathcal{S}$ can be decomposed as
\[
X=\left\{\bigoplus_i
 \left[(X_i^{(\mathbb C)}\otimes P_{i,+}^{(\mathbb C)})
 \oplus(\overline{X_i^{(\mathbb C)}}\otimes P_{i,-}^{(\mathbb C)})\right]\right\}
\oplus\left\{\bigoplus_j X_j^{(\mathbb R)}\otimes P_j^{(\mathbb R)}\right\}
\oplus\left\{\bigoplus_k X_k^{(\mathbb H)}\otimes P_k^{(\mathbb H)}\right\},
\]
where
\[
X_i^{(\mathbb C)}\in\mathbb M_{n_i}(\mathbb C)_h,\qquad
X_j^{(\mathbb R)}\in\mathbb M_{n_j}(\mathbb R)_h,\qquad
X_k^{(\mathbb H)}\in\mathbb M_{n_k}(\mathbb H)_h.
\]
The matrices $P_{i,+}^{(\mathbb C)}$, $P_{i,-}^{(\mathbb C)}$,
$P_j^{(\mathbb R)}$, and $P_k^{(\mathbb H)}$ are diagonal with
strictly positive eigenvalues and have sizes $p_i$, $q_i$, $m_j$,
and $m_k$, respectively. They do not depend on $X$; summands of
multiplicity zero are omitted. The two complex-type coefficient
blocks are linked by conjugation.
\end{thm}

\begin{proof}
Let $\gamma:\mathcal B(\mathcal H_{\mathcal S})\to\mathcal A$
be a sufficient real conditional expectation, and put
$\mathcal J=\mathcal A\cap\mathcal B_h(\mathcal H_{\mathcal S})$.
Its Hermitian restriction $\gamma_h$ is a sufficient positive
projection onto $\mathcal J$, faithful by Lemma \ref{lem:faithful}.
By Theorem \ref{thm:pos2J}, it is a Jordan conditional expectation.
Let $\Pi_J$ denote the real Hilbert--Schmidt orthogonal projection
onto $\mathcal J$, and put $\omega=\gamma_h^*(I)$.
The Jordan part of Lemma \ref{lem:omega_proj} gives
\[
\omega>0,\qquad [\omega,a]=0\quad(a\in\mathcal J),\qquad
X=\omega\Pi_J(X)\quad(X\in\mathcal S).
\]
To prove $\omega\in\mathcal A'$, only the simple factors
$\mathbb M_1(\mathbb C)$ and $\mathbb M_1(\mathbb H)$ need attention:
all other simple real matrix factors are generated by their Hermitian
parts. For a skew Hermitian $A\in\mathcal A$ supported on such a
factor and $X\in\mathcal S$, put $H=[X,A]/2$. Sufficiency and left
module property give, for Hermitian $B$,
\[
\Tr H\gamma_h(B)=\re\Tr XA\gamma(B)
=\re\Tr X\gamma(AB)=\re\Tr XAB=\Tr HB.
\]
The trace argument in Lemma \ref{lem:omega_proj} thus also gives
$H=\omega\Pi_J(H)$. On this factor $\mathcal J$ is scalar, so
$H=c\omega$ there. Since $\Tr H=0$ and $\omega>0$, $c=0$.
Some $X$ is a nonzero scalar multiple of $\omega$ on that factor,
by the definition of $\mathcal H_{\mathcal S}$. Hence $[\omega,A]=0$,
which proves the claim.

In the representation (\ref{eq:real_decomp}), the commutant
therefore gives
\[
\omega=\left\{\bigoplus_i
 [(I_{n_i}\otimes P_{i,+}^{(\mathbb C)})
 \oplus(I_{n_i}\otimes P_{i,-}^{(\mathbb C)})]\right\}
\oplus\left\{\bigoplus_j I_{n_j}\otimes P_j^{(\mathbb R)}\right\}
\oplus\left\{\bigoplus_k I_{2n_k}\otimes P_k^{(\mathbb H)}\right\},
\]
initially with positive definite complex Hermitian multiplicity
matrices. After a unitary change of basis on each multiplicity
space, these matrices are diagonal with strictly positive eigenvalues.
Finally, $X\omega^{-1}=\Pi_J(X)\in\mathcal J$ has the linked
coefficient blocks specified by (\ref{eq:real_decomp}). Multiplication
by $\omega$ proves the asserted decomposition.
\end{proof}

We next prove Theorem \ref{thm:KI_deco} for real Jordan algebras.
In (\ref{eq:jordan_deco}), we identify one-dimensional factors with
$\mathbb{R}$ and split $\Gamma_1\cong\mathbb{R}\oplus\mathbb{R}$
into scalar factors. The possible irreducible unital Jordan
representations of the resulting simple factors are as follows,
up to unitary equivalence \cite{jordan3}:
\[
\begin{array}{c|c|c}
J_{\alpha}&\text{irreducible representations}&d_{\alpha,s}\\ \hline
\mathbb{R}&\text{scalar}&1\\
\mathbb{M}_{n}(\mathbb{R})_{h},\ n\ge2&\text{standard}&n\\
\mathbb{M}_{n}(\mathbb{C})_{h},\ n\ge2&a\mapsto a,\quad a\mapsto a^{\mathsf T}&n\\
\mathbb{M}_{n}(\mathbb{H})_{h},\ n\ge3&\text{standard complex realization}&2n\\
\Gamma_{2k},\ k\ge1&\text{one Clifford representation}&2^{k}\\
\Gamma_{2k+1},\ k\ge1&\text{two Clifford representations}&2^{k}
\end{array}
\]
Here irreducibility means that the image generates the full complex
matrix algebra. In particular,
$\mathbb{M}_2(\mathbb{H})_h\cong\Gamma_5$ has two representation
types, each of complex dimension four. Only the types occurring
in the given representation belong to $\Lambda_\alpha$.
For a complex Hermitian factor, the two possible coefficient blocks
are $x_\alpha(X)$ and $x_\alpha(X)^{\mathsf T}$.

\begin{proof}[Proof of Theorem \ref{thm:KI_deco}]
Put $\mathcal C=C^*(\mathcal A_J)$. The units of the simple
Jordan factors are central in $\mathcal C$. Decomposing the
corresponding summands of $\mathcal C$ into irreducible matrix
representations gives a unitary $U$ with
\[
U\mathcal CU^*=\bigoplus_\alpha\bigoplus_{s\in\Lambda_\alpha}
\bigl(\mathbb M_{d_{\alpha,s}}(\mathbb C)\otimes I_{m_{\alpha,s}}\bigr).
\]
Restriction to $\mathcal A_J$ gives the representations
$\pi_{\alpha,s}$ and the asserted formula for $UaU^*$.
These restrictions remain irreducible and pairwise inequivalent,
since $\mathcal A_J$ generates $\mathcal C$.

By Lemmas \ref{lem:omega} and \ref{lem:omega_proj}, a sufficient
Jordan conditional expectation yields $\omega>0$ commuting with
$\mathcal C$ and $X=\omega y_X$, where $y_X\in\mathcal A_J$.
Thus
\[
U\omega U^*=\bigoplus_\alpha\bigoplus_{s\in\Lambda_\alpha}
I_{d_{\alpha,s}}\otimes P_{\alpha,s},\qquad P_{\alpha,s}>0.
\]
Unitary changes of basis on the multiplicity spaces make every
$P_{\alpha,s}$ diagonal without changing the representations
$\pi_{\alpha,s}$. Setting $x_\alpha(X)=\theta(y_X)_\alpha$ and
multiplying the decompositions of $\omega$ and $y_X$ proves the
formula for $UXU^*$.
\end{proof}

\section{Sufficient support size of measurements for quantum estimation\label{sec:support}}

An important application of sufficient real Jordan algebras is the
problem of sufficient support size of measurements for quantum estimation
\cite{support}. Let $\mathcal{S}(\mathcal{H})$ denote the set of
all density operators on a Hilbert space $\mathcal{H}$. Consider
a smooth family of quantum states $\{\rho_{\theta}\in\mathcal{S}(\mathcal{H})\mid\theta\in\Theta\subset\mathbb{R}^{d}\},$
and let $M=\{M_{x}\}^{s}_{x=1}$ be a POVM with $M_{x}\in\mathcal{B}(\mathcal{H})_{+}$
and $\sum^{s}_{x=1}M_{x}=I$. The corresponding classical Fisher information
matrix is defined by
\[
J^{(M)}_{\theta,ij}=\sum_{\substack{1\le x\le s\\\Tr\rho_{\theta}M_x>0}}\frac{\bigl(\Tr\partial_{i}\rho_{\theta}\,M_{x}\bigr)\bigl(\Tr\partial_{j}\rho_{\theta}\,M_{x}\bigr)}{\Tr\rho_{\theta}M_{x}}.
\]
For local estimation at an interior point $\theta_{0}\in\Theta$,
we consider the infimum of
\[
\tr W_{\theta_{0}}\bigl(J^{(M)}_{\theta_{0}}\bigr)^{-1}
\]
over POVMs, assigning the value $+\infty$ when $J^{(M)}_{\theta_0}$
is singular \cite{helstrom}.
Here $W_{\theta_{0}}$ is a given $d\times d$ real positive definite
matrix. The infimum need not be attained when $\rho_{\theta_0}$ is
singular. In general, however, such an optimization can only be carried
out numerically. Since the support size $s$ of the POVM is a priori
unbounded, it is impossible to optimize over all $s$. In \cite{support},
it was shown that this infimum is unchanged by restricting to POVMs satisfying
\[
s\le\dim\mathcal{A}_{J}+\frac{1}{2}d(d+1)-1,
\]
where $\mathcal{A}_{J}$ is a sufficient real Jordan algebra for the
local model $\mathcal{S}=\{\rho_{\theta_{0}}\}\cup\{\partial_{i}\rho_{\theta_{0}}\}^{d}_{i=1}$.
The abstract Jordan decomposition (\ref{eq:jordan_deco}) gives
\[
\dim\mathcal{A}_{J}=\sum_{i}n^{2}_{i}+\sum_{j}\frac{n_{j}(n_{j}+1)}{2}+\sum_{k}(2n^{2}_{k}-n_{k})+\sum_{l}(n_{l}+1).
\]

In Bayesian estimation, one is given a prior distribution $\pi:\Theta\to\mathbb{R}_{+}$,
and a pair consisting of a POVM $M=\{M_{x}\}^{s}_{x=1}$ and a map
$\hat{\theta}:\{1,2,\dots,s\}\to\mathbb{R}^{d}$ is regarded as an
estimator. One then seeks to minimize
\[
\int_{\Theta}d\theta\,\pi(\theta)\sum_{i,j,x}W_{\theta,ij}\bigl(\hat{\theta}^{i}(x)-\theta^{i}\bigr)\bigl(\hat{\theta}^{j}(x)-\theta^{j}\bigr)\Tr\rho_{\theta}M_{x}.
\]
Again, since the support size $s$ of the POVM is unbounded, it is
impossible to perform the optimization over all $s$. In \cite{support},
it was shown that the infimum is unchanged by restricting to
POVMs satisfying
\[
s\le\dim\mathcal{A}_{J},
\]
where $\mathcal{A}_{J}$ is a sufficient real Jordan algebra for the
model $\mathcal{S}=\{\rho_{\theta}\in\mathcal{S}(\mathcal{H})\mid\theta\in\Theta\subset\mathbb{R}^{d}\}.$

\section{Examples\label{sec:examples}}

In this section, we illustrate the strength of the present theory
of sufficient algebras through several examples.
\begin{example}[One-parameter pure-state local model]
Consider the one-parameter family of pure states on a two-dimensional
Hilbert space
\[
\left\{ \rho_{\theta}=\frac{1}{2}\bigl(I+\cos\theta\,Z+\sin\theta\,X\bigr)\ |\ -a<\theta<a\right\} ,
\]
where $a>0$ is sufficiently small, and
\[
X=\begin{pmatrix}0 & 1\\
1 & 0
\end{pmatrix},\qquad Z=\begin{pmatrix}1 & 0\\
0 & -1
\end{pmatrix}.
\]
We consider the local model at $\theta=0$, $\mathcal{S}=\{\rho_{0},\partial\rho_{0}\}.$
Here
\[
\rho=\rho_{0}=\ket 0\bra 0,\qquad\ket 0=\begin{pmatrix}1\\
0
\end{pmatrix},\qquad\partial\rho_{0}=\frac{1}{2}X.
\]
The SLD associated with $\partial\rho_{0}$ is $L=X$. The real $*$-subalgebra
generated by $X$ and $\rho$ is $\mathbb{M}_{2}(\mathbb{R})$, and
it satisfies
\[
\rho\,\mathbb{M}_{2}(\mathbb{R})\,\rho^{-1}\subset\mathbb{M}_{2}(\mathbb{R}).
\]
Hence, by Theorem \ref{thm:min_suff_alg}, the minimal sufficient
real $*$-subalgebra is $\mathbb{M}_{2}(\mathbb{R})$.
Here the normalized reference in Theorem \ref{thm:min_suff_J}
is $\rho_1=\rho$. Thus, by that theorem, the real Jordan algebra generated by
$\rho$ and $X$ is $\mathbb{M}_{2}(\mathbb{R})_{h}$, which is therefore
the minimal sufficient real Jordan algebra.

More generally, consider an arbitrary smooth one-parameter family
of pure states $\{\rho_{\theta}\mid-a<\theta<a\}$ on a finite-dimensional
Hilbert space, and the corresponding local model $\mathcal{S}=\{\rho_{0},\partial\rho_{0}\}.$
Then
\[
L=2\partial\rho_{0}
\]
is the SLD associated with $\partial\rho_{0}$. Writing
\[
\rho_{0}=\ket{\psi}\bra{\psi},\qquad\ket l=L\ket{\psi},
\]
we have
\[
L=\ket{\psi}\bra l+\ket l\bra{\psi},
\]
and $\ket{\psi}$ and $\ket l$ are orthogonal. If $\ket l\ne0$,
the basis $\{\ket{\psi},\ket l/\|l\|\}$ of
$\mathcal{H}_{\mathcal{S}}=\Span_{\mathbb C}\{\ket{\psi},\ket l\}$
makes $\rho_0$ and $L$ real matrices. They generate
$\mathbb{M}_{2}(\mathbb{R})$ as a real $*$-algebra and
$\mathbb{M}_{2}(\mathbb{R})_h$ as a unital Jordan algebra;
these are the respective minimal sufficient algebras.
If $\ket l=0$, both algebras reduce to $\mathbb{R}$ on the
one-dimensional support space. The same structure persists for
the local model of $\rho^{\otimes n}_{\theta}$ after restriction
to its support space. 
\end{example}

\begin{example}[Two-parameter pure-state local model]
 Consider a smooth two-parameter family of quantum states $\{\rho_{\theta}\mid\theta\in\Theta\subset\mathbb{R}^{2}\}$
on a finite-dimensional Hilbert space $\mathcal{H}$, and fix $\theta_{0}\in\Theta$.
We consider the local model
\[
\mathcal{S}=\{\rho_{\theta_{0}},\partial_{1}\rho_{\theta_{0}},\partial_{2}\rho_{\theta_{0}}\}.
\]
If $\rho_{\theta_{0}}$ is pure, then
\[
L_{i}=2\partial_{i}\rho_{\theta_{0}},\qquad i=1,2
\]
are the SLDs associated with $\partial_{i}\rho_{\theta_{0}}$. Writing
\[
\rho_{\theta_{0}}=\ket{\psi}\bra{\psi},\qquad\ket{l_{i}}=L_{i}\ket{\psi},
\]
we have
\[
L_{i}=\ket{\psi}\bra{l_{i}}+\ket{l_{i}}\bra{\psi},\qquad i=1,2.
\]
The vectors $\ket{\psi}$ and $\ket{l_{i}}$ are orthogonal.
Assume first that $\ket{l_1}$ and $\ket{l_2}$ are real-linearly
independent; otherwise the local model reduces to the preceding case.

If $\braket{l_{1}}{l_{2}}\in\mathbb{R}$, then $\rho_{\theta_{0}},L_{1},L_{2}$
are contained in a real $*$-subalgebra isomorphic to $\mathbb{M}_{3}(\mathbb{R})$
represented by real matrices in an orthonormal basis of
$\mathcal{H}_{\mathcal{S}}=\Span_{\mathbb C}\{\ket{\psi},\ket{l_1},\ket{l_2}\}$.
If $\ket{l_{1}}$ and $\ket{l_{2}}$ are complex-linearly dependent,
say $\ket{l_{1}}=c\ket{l_{2}}$ with $c\in\mathbb{C}$, then $\rho_{\theta_{0}},L_{1},L_{2}$
are contained in a real $*$-subalgebra isomorphic to $\mathbb{M}_{2}(\mathbb{C})$.
If $\ket{l_{1}}$ and $\ket{l_{2}}$ are complex-linearly independent
and $\braket{l_{1}}{l_{2}}\notin\mathbb{R}$, then $\rho_{\theta_{0}},L_{1},L_{2}$
are contained in a real $*$-subalgebra isomorphic to $\mathbb{M}_{3}(\mathbb{C})$.
Using $L_i\rho_{\theta_0}L_j=\ket{l_i}\bra{l_j}$, one checks
that the indicated algebras are generated by $\rho_{\theta_0},L_1,L_2$.
Thus, in the three cases above, the minimal sufficient real $*$-subalgebra is respectively $\mathbb{M}_{3}(\mathbb{R})$,
$\mathbb{M}_{2}(\mathbb{C})$, or $\mathbb{M}_{3}(\mathbb{C})$. The
corresponding minimal sufficient real Jordan algebra is respectively
$\mathbb{M}_{3}(\mathbb{R})_{h}$, $\mathbb{M}_{2}(\mathbb{C})_{h}$,
or $\mathbb{M}_{3}(\mathbb{C})_{h}$.

This structure is unchanged under i.i.d. extension. Indeed, for the
local model of
\[
\rho^{\otimes n}_{\theta_{0}+h/\sqrt{n}},\qquad h\in\mathbb{R}^{2},
\]
one obtains
\[
\ket{\psi^{(n)}}=\ket{\psi}^{\otimes n},
\]
and
\[
\ket{l^{(n)}_{i}}=\frac{1}{\sqrt{n}}\sum^{n}_{k=1}\ket{\psi}^{\otimes(k-1)}\otimes\ket{l_{i}}\otimes\ket{\psi}^{\otimes(n-k)},
\]
so that
\[
\braket{l^{(n)}_{i}}{l^{(n)}_{j}}=\braket{l_{i}}{l_{j}}.
\]
Hence the same classification persists for the local i.i.d. model
on its support space.

The same phenomenon persists for local models of pure states with
an arbitrary number of parameters. Matsumoto \cite{matsumoto_pure}
showed that, for pure-state models, the attainable lower bound for
the weighted mean-square error in local quantum estimation is unchanged
by i.i.d. extension under the $h/\sqrt{n}$ parametrization above. The present theory suggests that this is because
the underlying sufficient algebraic structure itself remains unchanged
under i.i.d. extension.
\end{example}

\begin{example}[Quantum exponential family]
Let $\mathcal{H}$ be a finite-dimensional Hilbert space, and let
$\rho$ be a density operator on $\mathcal{H}$, not necessarily strictly
positive. Let $\{L_{i}\}^{d}_{i=1}\subset\mathcal{B}_{h}(\mathcal{H})$
be self-adjoint operators satisfying
\[
\Tr\rho L_{i}=0,\qquad(I-s)L_{i}(I-s)=0,
\]
where $s=\supp\rho$. Consider the quantum statistical model
\[
\mathcal{S}=\left\{ \rho_{\theta}=e^{\frac{1}{2}(\theta^{i}L_{i}-\phi(\theta))}\,\rho\,e^{\frac{1}{2}(\theta^{i}L_{i}-\phi(\theta))}\ \mid\ \theta\in\mathbb{R}^{d}\right\} ,
\]
where $\phi(\theta)$ is chosen so that $\Tr\rho_{\theta}=1$. This
may be regarded as a quantum analogue of an exponential family. In
contrast to the classical case, the rank may be degenerate; for instance,
the family may consist entirely of pure states.

The model $\mathcal{S}$ is absolutely continuous in the sense of
the present paper, namely through the existence of square-root likelihood
ratios. Indeed, a square-root likelihood ratio of $\rho_{\theta}$
with respect to $\rho$ is
\[
e^{\frac{1}{2}(\theta^{i}L_{i}-\phi(\theta))}.
\]
When $\rho$ is singular, this need not be the minimal-norm choice.
Moreover, $\{L_i\}_{i=1}^d$ are the minimal-norm SLDs at $\theta=0$.
Theorem \ref{thm:min_suff_alg} applied to the local model shows
that its minimal sufficient real $*$-subalgebra is the smallest
$\mathcal{A}$ satisfying
\[
\{L_i\}_{i=1}^d\subset\mathcal{A},\qquad
\rho\mathcal{A}\rho^{-1}\subset\mathcal{A}.
\]
This is also minimal sufficient for the entire family: every sufficient
map for the family preserves its derivatives at zero, while the
$\rho$-preserving conditional expectation onto $\mathcal{A}$
preserves the family because the displayed exponentials lie in
$\mathcal{A}$. The local and global support spaces coincide, since
$(I-s)L_i(I-s)=0$.

Let $\rho_1$ be obtained by normalizing every irreducible diagonal
block of $\rho$ as in Theorem \ref{thm:min_suff_J}.
Then the unital Jordan algebra generated by
$\{L_i\}_{i=1}^d\cup\{\rho_1\}$ is minimal sufficient for the
local model. Its sufficient Jordan conditional expectation also
preserves the whole family, by (\ref{eq:JordanCE2}) and Jordan
functional calculus for the exponentials. It is therefore minimal
sufficient for the family as well. If this algebra is of type
$\mathbb{M}_n(\mathbb{R})_h$, $\mathbb{M}_n(\mathbb{H})_h$,
or a spin factor $\Gamma_n$, then one sees structures that are
invisible in the conventional theory of sufficient algebras based
only on complex $*$-subalgebras.

This example shows that our theory of square-root likelihood ratios
and sufficiency based on self-adjoint operators provides a natural
formulation of quantum exponential families, including those with
degenerate rank. By contrast, if one adopts cocycles as the analogue
of likelihood ratios in the conventional way, a natural characterization
of quantum exponential families appears to be out of reach.
\end{example}

\section{Conclusion}

In this work, we developed a theory of quantum sufficiency based on
real $*$-subalgebras and real Jordan algebras. In contrast to the
conventional formulation, which is centered on families of states
and Radon--Nikodym cocycles associated with faithful reference states,
our framework allows models consisting of general self-adjoint operators,
including derivatives of states. Within this framework, self-adjoint
likelihood-type objects, such as square-root likelihood ratios and
symmetric logarithmic derivatives, arise naturally as fundamental
quantities. This makes it possible to treat ordinary quantum statistical
models and local quantum statistical structures within a unified framework.

A central point of the present approach is that self-adjoint likelihood-type
objects, rather than Radon--Nikodym cocycles, play the role of noncommutative
analogues of classical likelihood ratios, while the genuinely quantum
aspect is encoded separately by the $\rho$-modular structure. This
viewpoint removes the need to assume that the reference state is strictly
positive and leads naturally to corresponding notions of sufficient
real $*$-subalgebras and sufficient real Jordan algebras. Within
this setting, we established characterizations in terms of real positive
maps, real and complex conditional expectations, minimal sufficient
structures, and Koashi--Imoto type decompositions.

These results suggest that real Jordan structure provides a natural
framework for the statistical aspect of quantum theory beyond the
conventional complex $*$-algebraic setting. Possible directions for
future work include extensions to asymptotic statistical theory, in
particular to local asymptotic analysis and q-LAN type frameworks,
as well as infinite-dimensional extensions, more general classes of
statistical models, and further applications to quantum estimation,
quantum testing, and asymptotic quantum statistics.

\section*{ACKNOWLEDGMENTS}

This work was supported by JSPS KAKENHI Grant Numbers JP23H01090 and
JP22K03466 and JST ERATO Grant Number JPMJER2402, Japan. The author
thanks Lauritz van Luijk for helpful correspondence and insightful
comments. 

\appendix

\section{An example showing that the restriction to $\mathcal{H}_{\mathcal{S}}$
is necessary\label{sec:ill_ex}}

Let $\{\rho_{\theta}\mid\theta\in\Theta\}\subset\mathcal{S}(\mathcal{H})$
be a family of quantum states on a Hilbert space $\mathcal{H}$, and
consider the model
\[
\begin{pmatrix}\rho_{\theta} & 0\\
0 & 0
\end{pmatrix}=\begin{pmatrix}1 & 0\\
0 & 0
\end{pmatrix}\otimes\rho_{\theta}.
\]
Assume that all $\rho_{\theta}$ commute and that $\rho_{\theta}>0$
for every $\theta\in\Theta$. Consider the $*$-subalgebra $\mathcal{A}=I_{2}\otimes\mathcal{B}(\mathcal{H}).$
Then the map
\[
\alpha:\begin{pmatrix}A_{1} & A_{3}\\
A_{2} & A_{4}
\end{pmatrix}\longmapsto\begin{pmatrix}A_{1} & 0\\
0 & A_{1}
\end{pmatrix}
\]
is a sufficient complex conditional expectation; its trace dual is
CPTP. It is not faithful on the enlarged Hilbert space.

Fix $\theta_{0}\in\Theta$. Then one finds that 
\[
\begin{pmatrix}(\rho_{\theta}\rho^{-1}_{\theta_{0}})^{1/2} & 0\\
0 & 0
\end{pmatrix}\notin\mathcal{A}.
\]
At first sight, this seems to contradict Theorem \ref{thm:ratio_abs},
which states that the likelihood ratio belongs to a sufficient algebra.
The point is that this pathology occurs because we have not restricted
the Hilbert space to $\mathcal{H}_{\mathcal{S}}$. In the present
example, the correct formulation should first pass to $\mathcal{H}_{\mathcal{S}}$.
Similar pathological examples may arise even in classical statistics.

\section{$\protect\D_{\rho}$-Invariance and $\rho$-Modular Invariance \label{sec:dinv}}

Let $\{\rho_{\theta}\mid\theta\in\Theta\subset\mathbb{R}^{d}\}$ be
a smoothly parameterized family of quantum states on a finite-dimensional
Hilbert space $\mathcal{H}$, and fix a point $\theta_{0}\in\Theta$.
Let $L_{1},L_{2},\dots,L_{d}$ be the symmetric logarithmic derivatives
(SLDs) at $\theta_{0}$, and write $\rho=\rho_{\theta_{0}}$. In the
preliminary discussion below, we assume that $\rho>0$, although in
many situations this assumption can be removed. 

Let $X_{1},X_{2},\dots,X_{r}$ be observables satisfying $L_{i}\in\mathcal{X}:=\Span_{\mathbb{C}}\{X_{k}\}^{r}_{k=1}\,(1\le i\le d),$
and $\mathcal{D}_{\rho}(\mathcal{X})\subset\mathcal{X}.$ Then $X_{1},X_{2},\dots,X_{r}$
are called a $\D_{\rho}$-invariant extension of the SLDs. Here $\ensuremath{\mathcal{D}_{\rho}:\mathcal{B}(\mathcal{H})\to\mathcal{B}(\mathcal{H})}$
is the linear map defined by
\begin{equation}
\D_{\rho}=\ii\frac{\mathcal{L_{\rho}}-\mathcal{R}_{\rho}}{\mathcal{L_{\rho}}+\mathcal{R}_{\rho}},\label{eq:d-ope}
\end{equation}
where $\mathcal{L}_{\rho}(Z)=\rho Z$, and $\mathcal{R}_{\rho}(Z)=Z\rho.$
Note that $\mathcal{L}_{\rho}$ and $\mathcal{R}_{\rho}$ commute.

In quantum estimation theory, $\D_{\rho}$-invariant extensions of
the SLDs often play an important role. For example, when estimating
$\theta$, a pair $(M,\hat{\theta})$ consisting of a POVM $M$ and
an estimator map $\hat{\theta}$ is regarded as an estimator. If $(M,\hat{\theta})$
is locally unbiased at $\theta_{0}\in\Theta$, then its covariance
matrix $V_{\theta_{0}}[M,\hat{\theta}]$ satisfies
\[
\tr WV_{\theta_{0}}[M,\hat{\theta}]\ge c^{(H)},
\]
where $W>0$ is a given real positive matrix. The quantity $c^{(H)}$
is called the Holevo Cramér--Rao bound, and it can be computed efficiently
using a $\D_{\rho}$-invariant extension of the SLDs \cite{qLAN1,holevo}.
Centering the $X_i$ so that $\Tr\rho X_i=0$ preserves the
invariant-extension property. Furthermore, in the asymptotic
representation theorem of \cite{qLAN3},
for the i.i.d. extension $\rho^{\otimes n}_{\theta}$ of the model,
the observables
\[
X^{(n)}_{i}=\frac{1}{\sqrt{n}}\sum^{n}_{k=1}I^{\otimes(k-1)}\otimes X_{i}\otimes I^{\otimes(n-k)}
\]
have the same asymptotic behavior as those of a quantum Gaussian state,
and this makes it possible to translate arbitrary sequences of statistics
into statistics on a quantum Gaussian model. In this q-LAN theory,
SLDs and square-root likelihood ratios play the roles of classical
logarithmic derivatives and likelihood ratios, while $\D_{\rho}$-invariant
extension plays an important auxiliary role. Note also that q-LAN
applies beyond the i.i.d. setting.

In connection with this, the following theorem relates $\D_{\rho}$-invariance
and $\rho$-modular invariance. Let $\rho\in\mathcal{S}(\mathcal{H})$
be a density operator, not necessarily faithful, and define $\D_{\rho}:\mathcal{B}(\mathcal{H})\to\mathcal{B}(\mathcal{H})$
by (\ref{eq:d-ope}), where the inverse of $\mathcal{L}_{\rho}+\mathcal{R}_{\rho}$
is understood as the Moore--Penrose generalized inverse with respect
to the Hilbert--Schmidt inner product. We also introduce
\begin{equation}
\widetilde{\D}_{\rho}=\frac{\mathcal{L_{\rho}}-\mathcal{R}_{\rho}}{\mathcal{L_{\rho}}+\mathcal{R}_{\rho}}.\label{eq:d-ope-1}
\end{equation}

\begin{thm}
Let $\mathcal{A}\subset\mathcal{B}(\mathcal{H})$ be a real $*$-subalgebra.
Then the following conditions are equivalent:
\begin{description}
\item [{(i)}] $\rho A\rho^{-1}\in\A$ for all $A\in\A$.
\item [{(ii)}] $\supp\rho\in\mathcal{A}$ and
$\widetilde{\mathcal{D}}_{\rho}(A)\in\A$ for all $A\in\A$.
\end{description}
Here $\rho^{-1}$ denotes the Moore--Penrose generalized inverse
of $\rho$.

If, moreover, $\mathcal{A}$ is a complex $*$-subalgebra, then these
conditions are also equivalent to
\begin{description}
\item [{(ii)'}] $\supp\rho\in\mathcal{A}$ and
$\D_{\rho}(A)\in\A$ for all $A\in\A$.
\end{description}
\end{thm}

\begin{proof}
We first prove (ii)$\Rightarrow$(i). Write $\rho=\begin{pmatrix}\rho_{1} & 0\\
0 & 0
\end{pmatrix}$, $\rho_{1}>0,$ and decompose $A=\begin{pmatrix}A_{1} & A_{2}\\
A_{3} & A_{4}
\end{pmatrix}$ accordingly. Then for every $n\in\mathbb{N}$,
\[
\widetilde{\mathcal{D}}^{n}_{\rho}(A)=\begin{pmatrix}\widetilde{\mathcal{D}}^{n}_{\rho_{1}}(A_{1}) & A_{2}\\
(-1)^{n}A_{3} & 0
\end{pmatrix}.
\]
If $\{p_{i}\}_{i}$ are the eigenvalues of $\rho_{1}$, then the eigenvalues
of $\widetilde{\mathcal{D}}_{\rho_{1}}$ are $\left\{ \frac{p_{i}-p_{j}}{p_{i}+p_{j}}\right\} _{i,j},$
so $\widetilde{\mathcal{D}}_{\rho_{1}}$ is a strict contraction.
Hence
\[
\lim_{n\to\infty}\widetilde{\mathcal{D}}^{2n}_{\rho}(A)=\begin{pmatrix}0 & A_{2}\\
A_{3} & 0
\end{pmatrix}\in\mathcal{A}.
\]
It follows that $\begin{pmatrix}A_{1} & 0\\
0 & A_{4}
\end{pmatrix}\in\mathcal{A}.$ Moreover,
\begin{equation}
\widetilde{\mathcal{D}}_{\rho_{1}}=\frac{\mathcal{L}_{\rho_{1}}-\mathcal{R}_{\rho_{1}}}{\mathcal{L}_{\rho_{1}}+\mathcal{R}_{\rho_{1}}}=\frac{\Delta_{\rho_{1}}-I}{\Delta_{\rho_{1}}+I},\qquad\Delta_{\rho_{1}}=\frac{I+\widetilde{\mathcal{D}}_{\rho_{1}}}{I-\widetilde{\mathcal{D}}_{\rho_{1}}},\label{eq:d2Delta1}
\end{equation}
where $\Delta_{\rho_{1}}(Z)=\rho_{1}Z\rho^{-1}_{1}$.
By (ii), $s=\supp\rho\in\mathcal{A}$, so
$sAs=\begin{pmatrix}A_1&0\\0&0\end{pmatrix}\in\mathcal{A}$.
Real polynomial functional calculus on $s\mathcal{A}s$ and
(\ref{eq:d2Delta1}) therefore give $\begin{pmatrix}\Delta_{\rho_{1}}(A_{1}) & 0\\
0 & 0
\end{pmatrix}=\rho A\rho^{-1}\in\mathcal{A}.$ Thus (ii)$\Rightarrow$(i).

Next we prove (i)$\Rightarrow$(ii). By Theorem \ref{thm:modular},
we have $s=\supp\rho\in\mathcal{A}$. Hence for $A=\begin{pmatrix}A_{1} & A_{2}\\
A_{3} & A_{4}
\end{pmatrix}\in\mathcal{A},$ we have
\[
sA\kappa=\begin{pmatrix}0 & A_{2}\\
0 & 0
\end{pmatrix}\in\mathcal{A},\qquad\kappa As=\begin{pmatrix}0 & 0\\
A_{3} & 0
\end{pmatrix}\in\mathcal{A},
\]
where $\kappa=I-s$. On the other hand,
\[
\rho A\rho^{-1}=\begin{pmatrix}\Delta_{\rho_{1}}(A_{1}) & 0\\
0 & 0
\end{pmatrix}\in\mathcal{A},
\]
so by (\ref{eq:d2Delta1}), $\begin{pmatrix}\widetilde{\mathcal{D}}_{\rho_{1}}(A_{1}) & 0\\
0 & 0
\end{pmatrix}\in\mathcal{A}.$ Therefore,
\[
\widetilde{\mathcal{D}}_{\rho}(A)=\begin{pmatrix}\widetilde{\mathcal{D}}_{\rho_{1}}(A_{1}) & A_{2}\\
-A_{3} & 0
\end{pmatrix}\in\mathcal{A},
\]
and (i)$\Rightarrow$(ii) follows.

Finally, if $\mathcal{A}$ is a complex $*$-subalgebra, then (ii)
and (ii)' are equivalent, since $\mathcal{D}_{\rho}=\ii\,\widetilde{\mathcal{D}}_{\rho}.$
\end{proof}

\end{document}